\documentclass[11pt]{article}

\usepackage{times}

\usepackage[paperwidth=199.8mm,paperheight=297mm,centering,hmargin=25mm,vmargin=25mm]{geometry}
\usepackage{authblk} 
\usepackage[titletoc,title]{appendix}
\usepackage{hyperref}
\hypersetup{colorlinks=true,citecolor=blue,linkcolor=blue,filecolor=blue,urlcolor=blue,breaklinks=true}
\usepackage[dvipsnames]{xcolor}
\usepackage{color}
\usepackage[marginal]{footmisc}
\usepackage{url}
\usepackage[noend]{algorithmic} 
\usepackage{bm}
\usepackage{thmtools} 
\usepackage{thm-restate}

\usepackage{mathtools}
\usepackage{amsmath}
\usepackage{amssymb}
\usepackage[shortlabels]{enumitem}
\usepackage{graphicx,epic,eepic,epsfig,latexsym,verbatim}
\usepackage{dsfont}

\usepackage{framed}
\definecolor{shadecolor}{rgb}{0.9,0.9,0.9}
\usepackage{subcaption}
\usepackage{caption}
\usepackage{float}
\usepackage{tikz}
\usepackage{tcolorbox}
\usepackage{relsize}

\definecolor{LightGray}{RGB}{220,220,220}
\definecolor{myred}{RGB}{236, 17, 0}
\definecolor{myblue}{RGB}{10, 88, 153}
\definecolor{myorange}{RGB}{236, 137, 0}
\definecolor{mygreen}{RGB}{26, 152, 81}

\usepackage{amsthm}
\newtheorem{definition}{Definition}
\newtheorem{proposition}[definition]{Proposition}
\newtheorem{lemma}[definition]{Lemma}

\def\squareforqed{\hbox{\rlap{$\sqcap$}$\sqcup$}}
\def\qed{\ifmmode\squareforqed\else{\unskip\nobreak\hfil
\penalty50\hskip1em\null\nobreak\hfil\squareforqed
\parfillskip=0pt\finalhyphendemerits=0\endgraf}\fi}
\def\endenv{\ifmmode\;\else{\unskip\nobreak\hfil
\penalty50\hskip1em\null\nobreak\hfil\;
\parfillskip=0pt\finalhyphendemerits=0\endgraf}\fi}

\newcounter{example}

\mathchardef\ordinarycolon\mathcode`\:
\mathcode`\:=\string"8000
\def\vcentcolon{\mathrel{\mathop\ordinarycolon}}
\begingroup \catcode`\:=\active
  \lowercase{\endgroup
  \let :\vcentcolon
  }

\usepackage{colortbl}
\usepackage{pifont}
\definecolor{Gray}{gray}{0.92}
\definecolor{Gray2}{gray}{0.75}
\definecolor{maroon}{cmyk}{0,0.87,0.68,0.32}
\usepackage{booktabs}
\usepackage{makecell}
\usepackage{diagbox}
\usepackage{multirow}

\newcommand{\nc}{\newcommand}
\nc{\rnc}{\renewcommand}
\nc{\bra}[1]{\langle#1|}
\nc{\ket}[1]{|#1\rangle}
\nc{\<}{\langle}
\rnc{\>}{\rangle}
\nc{\ketbra}[2]{|#1\rangle\!\langle#2|}
\nc{\braket}[2]{\langle#1|#2\rangle}
\nc{\braandket}[3]{\langle #1|#2|#3\rangle}
\nc{\proj}[1]{| #1\rangle\!\langle #1 |}
\nc{\avg}[1]{\langle#1\rangle}
\nc{\Rank}{\operatorname{Rank}}
\nc{\smfrac}[2]{\mbox{$\frac{#1}{#2}$}}
\nc{\tr}{\operatorname{Tr}}
\nc{\ox}{\otimes}

\nc{\cA}{{\cal A}}
\nc{\cB}{{\cal B}}
\nc{\cC}{{\cal C}}
\nc{\cD}{{\cal D}}
\nc{\cE}{{\cal E}}
\nc{\cF}{{\cal F}}
\nc{\cG}{{\cal G}}
\nc{\cH}{{\cal H}}
\nc{\cI}{{\cal I}}
\nc{\cJ}{{\cal J}}
\nc{\cK}{{\cal K}}
\nc{\cL}{{\cal L}}
\nc{\cM}{{\cal M}}
\nc{\cN}{{\cal N}}
\nc{\cO}{{\cal O}}
\nc{\cP}{{\cal P}}
\nc{\cQ}{{\cal Q}}
\nc{\cR}{{\cal R}}
\nc{\cS}{{\cal S}}
\nc{\cT}{{\cal T}}
\nc{\cV}{{\cal V}}
\nc{\cU}{{\cal U}}
\nc{\cX}{{\cal X}}
\nc{\cY}{{\cal Y}}
\nc{\cZ}{{\cal Z}}
\nc{\cW}{{\cal W}}

\nc{\RR}{{{\mathbb R}}}
\nc{\CC}{{{\mathbb C}}}
\nc{\FF}{{{\mathbb F}}}
\nc{\NN}{{{\mathbb N}}}
\nc{\ZZ}{{{\mathbb Z}}}
\nc{\QQ}{{{\mathbb Q}}}
\nc{\UU}{{{\mathbb U}}}
\nc{\EE}{{{\mathbb E}}}
\nc{\id}{{\operatorname{id}}}

\nc{\1}{{\mathds{1}}}
\nc{\supp}{{\operatorname{supp}}}

\def\ve{\varepsilon}

\makeatletter
\def\grd@save@target#1{%
  \def\grd@target{#1}}
\def\grd@save@start#1{%
  \def\grd@start{#1}}
\tikzset{
  grid with coordinates/.style={
    to path={%
      \pgfextra{%
        \edef\grd@@target{(\tikztotarget)}%
        \tikz@scan@one@point\grd@save@target\grd@@target\relax
        \edef\grd@@start{(\tikztostart)}%
        \tikz@scan@one@point\grd@save@start\grd@@start\relax
        \draw[minor help lines,magenta] (\tikztostart) grid (\tikztotarget);
        \draw[major help lines] (\tikztostart) grid (\tikztotarget);
        \grd@start
        \pgfmathsetmacro{\grd@xa}{\the\pgf@x/1cm}
        \pgfmathsetmacro{\grd@ya}{\the\pgf@y/1cm}
        \grd@target
        \pgfmathsetmacro{\grd@xb}{\the\pgf@x/1cm}
        \pgfmathsetmacro{\grd@yb}{\the\pgf@y/1cm}
        \pgfmathsetmacro{\grd@xc}{\grd@xa + \pgfkeysvalueof{/tikz/grid with coordinates/major step}}
        \pgfmathsetmacro{\grd@yc}{\grd@ya + \pgfkeysvalueof{/tikz/grid with coordinates/major step}}
        \foreach \x in {\grd@xa,\grd@xc,...,\grd@xb}
        \node[anchor=north] at (\x,\grd@ya) {\pgfmathprintnumber{\x}};
        \foreach \y in {\grd@ya,\grd@yc,...,\grd@yb}
        \node[anchor=east] at (\grd@xa,\y) {\pgfmathprintnumber{\y}};
      }
    }
  },
  minor help lines/.style={
    help lines,
    step=\pgfkeysvalueof{/tikz/grid with coordinates/minor step}
  },
  major help lines/.style={
    help lines,
    line width=\pgfkeysvalueof{/tikz/grid with coordinates/major line width},
    step=\pgfkeysvalueof{/tikz/grid with coordinates/major step}
  },
  grid with coordinates/.cd,
  minor step/.initial=.2,
  major step/.initial=1,
  major line width/.initial=2pt,
}
\makeatother

\makeatletter
\newcommand*\rel@kern[1]{\kern#1\dimexpr\macc@kerna}
\newcommand*\widebar[1]{%
  \begingroup
  \def\mathaccent##1##2{%
    \rel@kern{0.8}%
    \overline{\rel@kern{-0.8}\macc@nucleus\rel@kern{0.2}}%
    \rel@kern{-0.2}%
  }%
  \macc@depth\@ne
  \let\math@bgroup\@empty \let\math@egroup\macc@set@skewchar
  \mathsurround\z@ \frozen@everymath{\mathgroup\macc@group\relax}%
  \macc@set@skewchar\relax
  \let\mathaccentV\macc@nested@a
  \macc@nested@a\relax111{#1}%
  \endgroup
}
\makeatother

\newcommand{\Renyi}{R\'{e}nyi }

\usepackage[ruled]{algorithm2e}
\usepackage{pgfplots}
\pgfplotsset{width=8.6cm, compat=1.17} 
\usepgfplotslibrary{statistics}
\usepgfplotslibrary{colorbrewer} 
\usepgflibrary{plotmarks} 
\usepgflibrary[plotmarks] 
\usetikzlibrary{plotmarks} 
\usetikzlibrary[plotmarks] 
\usepackage{filecontents} 
\usepackage{calc} 
\expandafter\def\expandafter\UrlBreaks\expandafter{\UrlBreaks
  \do\a\do\b\do\c\do\d\do\e\do\f\do\g\do\h\do\i\do\j%
  \do\k\do\l\do\m\do\n\do\o\do\p\do\q\do\r\do\s\do\t%
  \do\u\do\v\do\w\do\x\do\y\do\z\do\A\do\B\do\C\do\D%
  \do\E\do\F\do\G\do\H\do\I\do\J\do\K\do\L\do\M\do\N%
  \do\O\do\P\do\Q\do\R\do\S\do\T\do\U\do\V\do\W\do\X%
  \do\Y\do\Z\do\0\do\1\do\2\do\3\do\4\do\5\do\6\do\7\do\8\do\9} 

\newcommand{\md}{\mathrm{d}}

\usepackage{qcircuit}

\begin{document}

\title{\Large \textbf{Estimating quantum relative entropies on quantum computers}}

\author[1,2]{Yuchen Lu}
\author[1]{Kun Fang \thanks{kunfang@cuhk.edu.cn}}

\affil[1]{\small School of Data Science, The Chinese University of Hong Kong, Shenzhen,\protect\\  Guangdong, 518172, China}

\affil[2]{\small Ecole Polytechnique Fédéderale de Lausanne (EPFL),\protect\\ CH-1015 Lausanne, Switzerland}

\maketitle

\begin{abstract}
Quantum relative entropy, a quantum generalization of the renowned Kullback-Leibler divergence, serves as a fundamental measure of the distinguishability between quantum states and plays a pivotal role in quantum information science. Despite its importance, efficiently estimating quantum relative entropy between two quantum states on quantum computers remains a significant challenge. In this work, we propose the first quantum algorithm for directly estimating quantum relative entropy and Petz \Renyi divergence from two unknown quantum states on quantum computers, addressing open problems highlighted in [Phys. Rev. A 109, 032431 (2024)] and [IEEE Trans. Inf. Theory 70, 5653–5680 (2024)]. Notably, the circuit size of our algorithm is at most $2n+1$ with $n$ being the number of qubits in the quantum states and it is directly applicable to distributed scenarios, where quantum states to be compared are hosted on cross-platform quantum computers. We prove that our loss function is operator-convex, ensuring that any local minimum is also a global minimum.
We validate the effectiveness of our method through numerical experiments and observe the absence of the barren plateau phenomenon. As an application, we employ our algorithm to investigate the superadditivity of quantum channel capacity. Numerical simulations reveal new examples of qubit channels exhibiting strict superadditivity of coherent information, highlighting the potential of quantum machine learning to address quantum-native problems.
\end{abstract}

\tableofcontents

\section{Introduction}

Relative entropy, also known as the Kullback-Leibler (KL) divergence~\cite{kullback1951information}, quantifies the difference between a model probability distribution and the true distribution. It is a fundamental concept in classical information theory and has widespread applications across various fields. For example, KL divergence serves as a standard loss function in machine learning~\cite{murphy2012machine}, including in Boltzmann Machines~\cite{ackley1985learning}, which were recognized with the 2024 Nobel Prize in Physics for their foundational impact. Additionally, KL divergence is used to assess the efficiency of data compression schemes~\cite{cover1999elements,gray2011entropy} and to measure deviations from thermodynamic equilibrium in statistical mechanics~\cite{jaynes1957information,crooks1999entropy}.

Building on advances in quantum information science, several quantum generalizations of the KL divergence have been developed to quantify the distinguishability between quantum states \cite{umegaki1962conditional,belavkin1982c,donald1986relative,hiai1991proper}. Notably, the formulations introduced by Umegaki~\cite{umegaki1962conditional} and Petz~\cite{petz1986quasi} are particularly important due to their operational relevance in a wide range of quantum information processing tasks~\cite{hiai1991proper,mosonyi2011quantum}. These measures are widely used in quantum machine learning~\cite{biamonte2017quantum}, quantum channel coding~\cite{hayashi2016quantum}, quantum error correction~\cite{cerf1997information}, quantum resource theories~\cite{chitambar2019quantum}, and quantum cryptography~\cite{pirandola2020advances}.

Despite its broad utility, efficiently estimating quantum relative entropy remains a significant challenge. On classical computers, calculating the quantum relative entropy between two unknown quantum states requires explicit matrix representations of the density operators. However, reconstructing these matrices via quantum state tomography demands an exponential number of measurements with respect to the system dimension~\cite{paris2004quantum,kliesch2021theory}. Even after obtaining the density matrices, classical computation of quantum relative entropy involves manipulating exponentially large matrices, resulting in prohibitive computational costs. In line with Feynman's vision for quantum simulation~\cite{Feynman}, directly estimating quantum relative entropy on quantum computers—without reconstructing the density matrices—offers a promising and potentially more efficient alternative. Nevertheless, the nonlinear nature of quantum relative entropy poses substantial challenges for such an approach.

\subsection{Main contributions}

In this work, we propose the first quantum algorithm for directly estimating quantum relative entropy and Petz \Renyi divergence between two unknown quantum states on quantum computers, addressing open problems highlighted by Goldfeld et al.~\cite{goldfeld2024quantum} and Wang et al.~\cite{wang2024new}. Our approach employs quadrature approximations of relative entropies, the variational representation of quantum $f$-divergences, and a novel technique for parameterizing Hermitian polynomial operators to estimate their traces with quantum states. An overview is provided in Figure~\ref{fig:illustrative_Fig}. Specifically, we develop a variational quantum algorithm (VQA) to evaluate quantum $f$-divergences within a classical-quantum hybrid framework: probability distributions are sampled on quantum computers using quantum inputs, followed by classical post-processing. The evaluation of quantum $f$-divergences serves as a key subroutine for estimating quantum relative entropy and Petz \Renyi divergence via quadrature approximation, efficiently expressing these quantities as weighted averages of quantum $f$-divergences.

\begin{figure}[H]
    \centering
    \includegraphics[width=\linewidth]{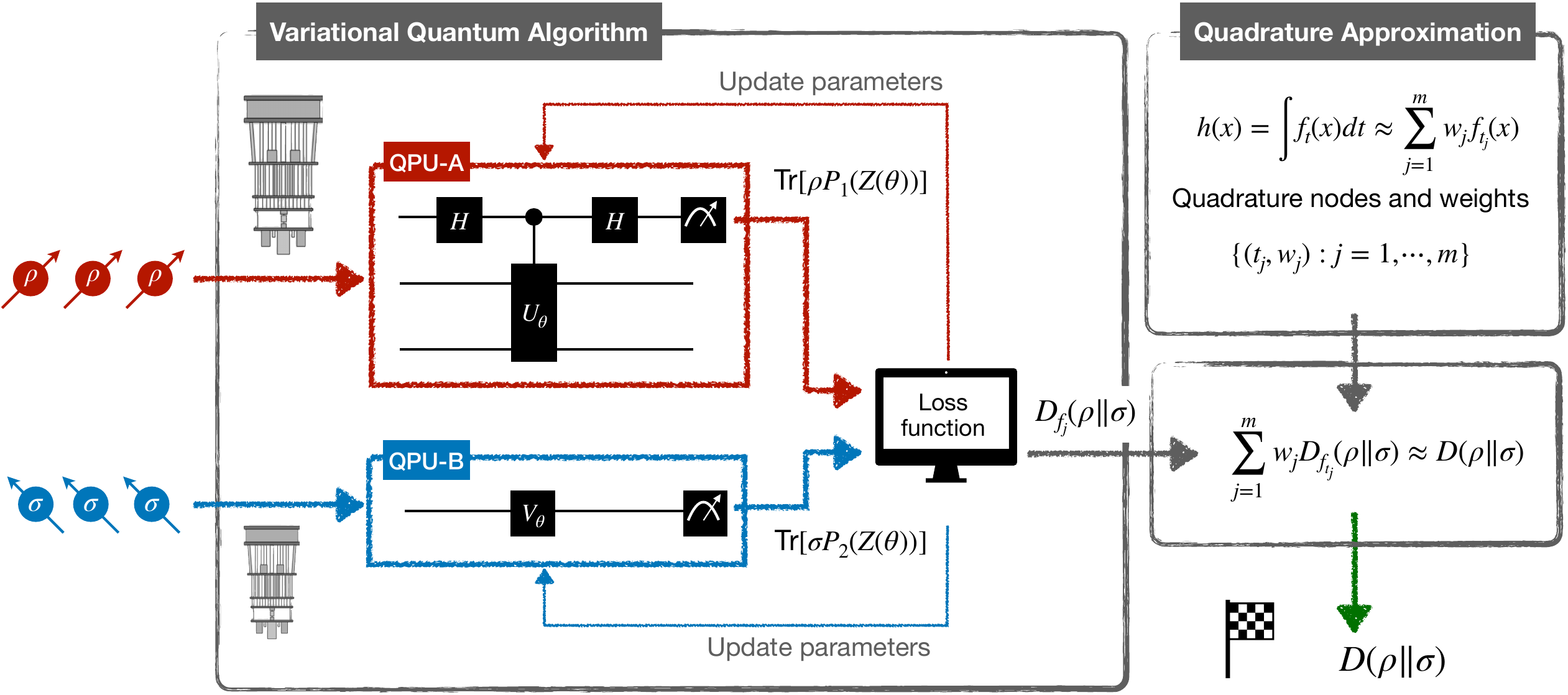}
    \caption{Estimating quantum relative entropy using a variational quantum algorithm. Identical copies of the quantum states, possibly hosted on different quantum computers, are processed through parameterized quantum circuits to extract relevant information. The variational algorithm estimates quantum $f$-divergences, and the quantum relative entropy is computed as a weighted average of these $f$-divergences, determined by the quadrature nodes and weights. The estimation of Petz \Renyi divergence follows a similar procedure.}
    \label{fig:illustrative_Fig}
\end{figure}

Notably, the circuit size of our algorithm is at most $2n+1$, where $n$ is the number of qubits in the quantum states. It is also directly applicable to distributed scenarios, enabling the comparison of quantum states hosted on cross-platform quantum computers. This feature is particularly relevant for cross-platform quantum verification, a topic of growing interest in recent studies~\cite{knorzer2023cross, greganti2021cross, zheng2024cross, elben2020cross, Qian2024prl}.
This advantage stems from the variational representation of $f$-divergence, which decouples the contributions of $\rho$ and $\sigma$, enabling independent estimation of terms involving $\rho$ and $\sigma$ on separate quantum devices, with the results subsequently combined to compute their relative entropies. 

We prove that our loss function is operator-convex, ensuring a well-behaved optimization landscape and the existence of efficient solutions. Numerical simulations further validate our approach, demonstrating that the parameterized quantum circuits remain trainable using simple optimizers such as gradient descent. By employing an adaptive learning rate strategy that reduces the step size in response to loss fluctuations, we achieve reliable convergence and obtain estimates with relative errors typically on the order of a few percent ($\approx 2\%$). Additionally, our numerical experiments indicate that the proposed loss function avoids the barren plateau phenomenon—a common challenge in VQAs where gradients vanish exponentially with increasing qubit count and circuit expressibility~\cite{mcclean2018barren,holmes2022connecting}. 

As an application, we utilize our algorithm to explore the superadditivity of quantum channel capacity—a phenomenon where the combined use of quantum channels exceeds the sum of their individual capacities. Specifically, we propose a variational approach to compute the quantum channel coherent information, employing our algorithm for quantum relative entropy as a subroutine. Numerical simulations uncover new examples of qubit channels exhibiting strict superadditivity of coherent information. Compared to methods based on classical neural networks and optimizations~\cite{leditzky2018dephrasure,bausch2020quantum, sidhardh2022exploring}, our approach leverages quantum systems to detect superadditivity. This demonstrates the potential of quantum machine learning to tackle quantum-native problems and paving the way for advancing quantum science through quantum computation.

\subsection{Related works}

Estimating quantum entropies and divergences is a central problem in quantum information science, and has attracted significant attention in recent years. Table~\ref{tab:related_works} provides an overview of related works and summarizes their main approaches.

\setlength\extrarowheight{2pt}
\begin{table}[ht]
\centering
\small
\begin{tabular}{c|c|c|c}
\toprule[2pt]
\textbf{Quantities} & \textbf{QSVT-based} & \textbf{VQA} & \textbf{Others} \\ \hline
 Von Neumann entropy              &\cite{gilyen2019distributional,wang2024new}          &      \cite{shin2024estimating,goldfeld2024quantum} & \cite{acharya2020estimating, wang2023quantum}  \\ \hline
Quantum ($\alpha$-)\Renyi entropy                     & \cite{subramanian2021quantum,wang2024new}          &   \cite{goldfeld2024quantum}  & \cite{acharya2020estimating, wang2023quantum} \\ \hline
Quantum ($\alpha$-)Tsallis entropy                    & \cite{wang2024new}           &   $\backslash$  & $\backslash$  \\ \hline
  ($\alpha$-)trace distance                       & \cite{wang2024new}           &      \cite{chen2021variational}  &\cite{buadescu2019quantum}  \\ \hline
  ($\alpha$-)fidelity                      &\cite{wang2024new}           &    \cite{cerezo2020variational,chen2021variational}  & \cite{buadescu2019quantum}   \\ \hline
\rowcolor{White!50!lightgray}  Quantum relative entropy                      & $\backslash$           &           this work   & $\backslash$  \\ \hline
 \rowcolor{White!50!lightgray} Petz \Renyi divergence                       & $\backslash$           &             this work  & $\backslash$  \\ 
\bottomrule[2pt] 
\end{tabular}
\caption{Related works of estimation methods for different quantum quantities.}
\label{tab:related_works}
\end{table}

The first class of algorithms involves the ``purified quantum query access'', a widely adopted input model based on the general framework of quantum singular value transformation (QSVT) \cite{gilyen2019quantum}. In this model, mixed quantum states are represented through quantum oracles that prepare their purifications. For instance,~\cite{gilyen2019distributional} introduced a quantum algorithm for computing von Neumann entropy with query complexity $O(d/\varepsilon^{1.5})$. Additionally,~\cite{subramanian2021quantum} proposed a method of computing the quantum \Renyi entropy with query complexity $O\left(\frac{\kappa}{(x\varepsilon)^2}\log\left(\frac{d}{\varepsilon}\right)\right)$, where $\kappa > 0$ satisfies $I/\kappa \le \rho \le I$ and $x = \tr[\rho^\alpha]/d$. Expanding on this framework,~\cite{wang2024new} provided a series of new quantum algorithms that can compute quantum entropies and quantum divergences such as the quantum \Renyi entropy, quantum Tsallis entropy, trace distance, and $\alpha$-fidelity for $\alpha \in (0,1)$. While the quantum relative entropy can be approached by the Sandwiched \Renyi divergence in principle as $\alpha\to 1$, the latter of which can be computed from $\alpha$-fidelity using the algorithm provided in~\cite{wang2024new}. However, the query complexity of this approach is exponential in $1/(1-\alpha)$, which blows up as $\alpha \to 1$.

The second class of algorithms involves estimating quantum entropies using VQAs, which have proven effective for tackling problems that are challenging to solve exactly~\cite{cerezo2021variational}. For instance, ~\cite{cerezo2020variational} introduced a hybrid classical-quantum algorithm to estimate truncated fidelity. Similarly, ~\cite{chen2021variational} proposed VQAs for estimating trace distance and fidelity by reformulating these tasks as optimization problems over unitary operators. Building on this, ~\cite{shin2024estimating} developed a variational method for estimating von Neumann entropy by expressing its variational formula as an optimization problem over parameterized quantum states. Additionally, ~\cite{goldfeld2024quantum} presented a suite of VQAs for estimating von Neumann and \Renyi entropies, as well as measured relative entropy and measured \Renyi relative entropy, by parameterizing Hermitian operators using parameterized quantum circuits and classical neural networks.

Besides, \cite{buadescu2019quantum} provided two algorithms for testing the closeness of quantum states with respect to the fidelity and trace distance. 
\cite{acharya2020estimating} introduced a method to compute the von Neumann entropy and quantum \Renyi entropies of an $d$-dimensional quantum state with $O(d^2/\varepsilon^2)$ and $O(d^{2/\alpha}/\varepsilon^{2/\alpha})$ copies, respectively. \cite{wang2023quantum} exploited truncated Taylor series to estimate von Neumann entropy and quantum \Renyi entropies and showed that the corresponding quantum circuits can be efficiently constructed with single and two-qubit quantum gates. 
 
Despite these advancements, all current methods fail to estimate quantum relative entropy and Petz \Renyi divergences from unknown quantum states due to various inherent limitations.

\subsection{Organization}

The remainder of this paper is organized as follows. Section~\ref{sec:pre} introduces the variational expressions for quantum relative entropy and Petz \Renyi divergence. Section~\ref{sec:polynomial_Z} presents a technique for parameterizing Hermitian polynomial operators, enabling the estimation of their traces with quantum states. Building on these foundations, Section~\ref{sec:VQAs} details variational quantum algorithms for quantum $f$-divergence, quantum relative entropy, and Petz \Renyi divergence. Section~\ref{sec:converge_analysis} analyzes the convergence properties and trainability of our approach. Section~\ref{sec:error_analysis} provides an error analysis of the proposed algorithms. Section~\ref{sec:numerical_simu} reports the results of numerical simulations. Section~\ref{sec:applications} demonstrates an application of our method to quantum channel capacity. Finally, Section~\ref{sec:conclusion} concludes with a summary and discusses open questions for future research.

\section{Preliminaries}
\label{sec:pre}

\subsection{Notations}

In this section, we introduce the notations and define key quantities used throughout this work. Table~\ref{tab:notations} summarizes the most frequently used symbols. A quantum state refers to a positive semidefinite operator with unit trace. We denote quantum states by $\rho$ and $\sigma$, and their supports by $\supp(\rho)$ and $\supp(\sigma)$, respectively. The logarithm with base 2 is written as \(\log\), while the natural logarithm (base \(e\)) is denoted by \(\ln\). The set of integers \(\{1, 2, \dots, p\}\) is represented by \([p]\). Estimated values are indicated by a hat, e.g., \(\hat{x}\).

\setlength\extrarowheight{3pt}
\begin{table}[ht]
\centering
\begin{tabular}{l|l}
\toprule[2pt]
\multicolumn{1}{c|}{\textbf{Notations}} & \multicolumn{1}{c}{\textbf{Descriptions}} \\ \hline
$D(\rho\|\sigma)$          & Quantum relative entropy          \\ \hline
$D_\alpha(\rho\|\sigma)$   & Petz  \Renyi divergence with $\alpha\in (0,1)\cup(1,2]$ \\ \hline
$D_f(\rho\|\sigma)$        & Standard quantum $f$-divergence           \\ \hline
$Q_\alpha(\rho\|\sigma)$   & Quasi-relative entropy $\tr[\rho^\alpha \sigma^{1-\alpha}]$  \\ \hline
$\hat{x}$                  & Estimate of a quantity $x$, such as $\hat{D}$, $\hat{D}_\alpha$, $\hat{D}_f$, $\hat{Q}_\alpha$            \\ \hline
$f_t(x)$       &  Function $(x-1)/(t(x-1) + 1)$ with $t \in [0,1]$ \\ \hline
$r_m(x)$       &  Quadrature approximation for $\log(x)$ with $m$ quadrature nodes  \\ \hline
$h_m(x)$       &  Quadrature approximation for $x^{1-\alpha}$ with $m$ quadrature nodes  \\
\bottomrule[2pt] 
\end{tabular}
\caption{Notational conventions of divergences and quadrature approximation functions.}
\label{tab:notations}
\end{table}

\subsection{Variational expression of quantum relative entropies} 
\label{sec:vatiational_form}

In this section, we introduce the standard quantum $f$-divergence and present its variational expression for a specific class of $f_t$-divergences. We then derive variational forms for the quantum relative entropy and Petz \Renyi divergence, utilizing the Gauss-Radau-Jacobi (GRJ) quadrature approximation. The GRJ quadrature method efficiently approximates integrals with weight functions of the form $(1-t)^a(1+t)^b$ by discretizing the integral into a weighted sum over carefully selected quadrature nodes. By exploiting the linearity of the standard quantum $f$-divergence in $f$, we reformulate the integral representations of quantum relative entropy and Petz \Renyi divergence as sums of $f_t$-divergences. This approach enables these quantities to be efficiently estimated within the variational framework.

\subsubsection{Standard quantum $f$-divergence}

The definition of standard quantum $f$-divergence is as follows~\cite{petz1986quasi,hiai2017different}:
\begin{definition}[Standard quantum $f$-divergence]
    Let $\rho$, $\sigma$ be positive semidefinite operators on a finite-dimensional Hilbert space $\cH$ with spectral decompositions $\rho = \sum_{j}\eta_j P_j$ and $\sigma = \sum_{k}\mu_k Q_k$. The standard quantum $f$-divergence is defined by
    \begin{align}
        D_f(\rho\|\sigma) := \sum_{j:\eta_j > 0}\sum_{k:\mu_k>0} \eta_j f(\mu_k \eta_j^{-1}) \tr[P_jQ_k] + f(0^+)\tr[\rho(I-\sigma^0)],
        \label{eq:def:D_f}
    \end{align}
    where $\sigma^0$ denotes the projector onto the support of $\sigma$ and $f(0^+)$ denotes the right-hand limit of $f(x)$ as $x\to \infty$.
    \label{def:f_div}
\end{definition}
Note that for quantum states $\rho$ and $\sigma$ that satisfy $\supp(\rho)\subseteq \supp(\sigma)$, the second term in Eq.~\eqref{eq:def:D_f} vanishes since $\tr[\rho(I - \sigma^0)] = 0$. When we set the function $f$ in the definition above as
\begin{align}
    f_t(x) := \frac{x-1}{t(x-1) + 1}~~\text{for }t\in [0,1],
    \label{def:f_t(x)}
\end{align}
its corresponding standard divergence $D_{f_t}(\rho\|\sigma)$ has a variational expression.
\begin{lemma}[\cite{brown2024device}]
Let $\rho$ and $\sigma$ be positive semidefinite operators on a finite-dimensional Hilbert space $\cH$ and $t\in (0,1]$. Then
    \begin{align}
        D_{f_t}(\rho\|\sigma) = &\inf_{Z} \frac{1}{t}\left\{ \tr[\rho] + \tr[\rho(Z+Z^\dagger)] + (1-t)\tr[\rho Z^\dagger Z] + t\tr[\sigma Z Z^\dagger] \right\},
        \label{eq:variant_Df}
    \end{align}
    where the infimum on the right-hand side is taken over all bounded linear operators $Z$ on $\cH$.
    \label{lemma:variational_ft}
\end{lemma}

It is important to note that the variational expression above is not applicable at $t=0$. In this boundary case, the value can be directly computed using the definition of $f_0$-divergence (Eq.~\eqref{eq:def:D_f}), which yields $D_{f_0}(\rho\|\sigma) = \tr[\rho^0 \sigma] - 1$. For the remainder of this paper, we assume $\supp(\rho) = \supp(\sigma)$. In this case, we have $D_{f_0}(\rho\|\sigma) = 0$.

\subsubsection{Quantum relative entropy}

Let $\rho$ and $\sigma$ be two quantum states. Their quantum relative entropy is defined as~\cite{umegaki1962conditional}
\begin{align}
    D(\rho\|\sigma) := \tr[\rho(\log\rho - \log\sigma)],
\end{align}
if $\supp(\rho) \subseteq \supp(\sigma)$ and  $+ \infty$ otherwise. In particular, let $f(x) = -\log(x)$. The corresponding standard quantum $f$-divergence simplifies to the quantum relative entropy,
\begin{align}\label{eq: relative entropy log}
D_{-\log}(\rho\|\sigma)  = D(\rho\|\sigma).
\end{align}
With the integral representation of logarithm, $\ln(x) = \int_{0}^1 f_t(x) \md t$, we can approximate the logarithm using GRJ quadrature as follows~\cite[Theorem 3.8]{brown2024device},
\begin{align}
    \log(x)\approx r_m(x) :=\frac{1}{\ln 2}\sum_{j=1}^m w_j f_{t_j}(x),
    \label{eq:gaussian_qua}
\end{align}
where $\{t_j\}_{j=1}^m$ is the set of GRJ quadrature nodes with a fixed node $t_1 = 0$ or $t_m = 1$, and $\{w_j\}_{j=1}^m$ are corresponding weights. The quadrature approximation for $\log x$ is shown in Figure~\ref{fig:quadrature_approx} (a), which converges to the exact value as the number of quadrature nodes $m$ increases. Combining Eqs.~\eqref{eq: relative entropy log} and~\eqref{eq:gaussian_qua}, we obtain the approximation of quantum relative entropy as follows:
\begin{align}
    D(\rho\|\sigma) 
    \approx D_{-r_m}(\rho\|\sigma) = -\sum_{j=1}^m \frac{w_j}{\ln2} D_{f_{t_j}}(\rho\|\sigma) \label{eq:approx_rel_entr}
\end{align}
where each $D_{f_{t_j}}(\rho\|\sigma)$ on the right-hand side can be evaluated using the variational expression provided in Lemma~\ref{lemma:variational_ft}.

\begin{figure}[H]
    \centering
    \includegraphics[width=0.9\linewidth]{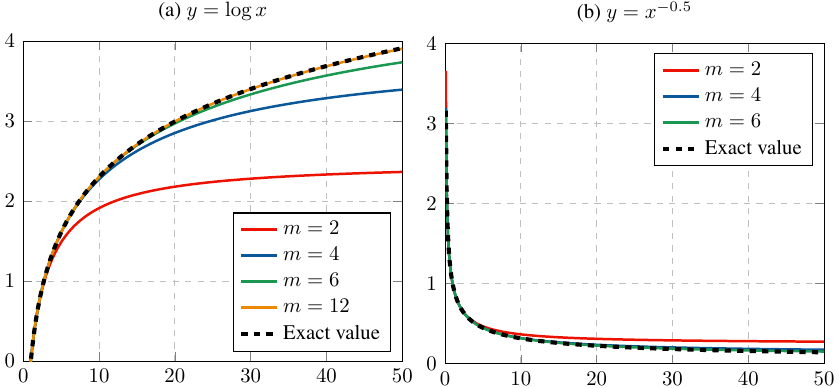}
    \caption{Performance of GRJ quadrature approximation. (a) Quadrature approximation for the function $y = \log x$; (b) Quadrature approximation for the function $y = x^{-0.5}$. The approximation converges to the exact value as the number of quadrature nodes $m$ increases.}
    \label{fig:quadrature_approx}
\end{figure}

\subsubsection{Petz \Renyi divergence} 

Let $\rho$ and $\sigma$ be two quantum states. The Petz \Renyi divergence is defined as~\cite{petz1986quasi}
\begin{align}
    D_\alpha (\rho\|\sigma) := \frac{1}{\alpha - 1}\log Q_\alpha(\rho\|\sigma),
\end{align}
if $\supp(\rho) \subseteq \supp(\sigma)$ and $+\infty$ otherwise. Here $Q_\alpha(\rho\|\sigma) := \tr[\rho^\alpha \sigma^{1-\alpha}]$ is the quasi-relative entropy, and we focus on $\alpha \in (0,1)\cup (1,2]$ which is the range of particular interest. Let $f(x) = x^{1-\alpha}$. The corresponding standard quantum $f$-divergence reduces to
\begin{align}
    D_{x^{1-\alpha}}(\rho\|\sigma) = Q_{\alpha}(\rho\|\sigma). \label{eq: power function f divergence}
\end{align} 
To implement the quadrature approximation, an integral representation for the power function is required, which is provided by the L\"{o}wner's Theorem~\cite{lowner1934monotone}, for $\alpha \in (0,1)\cup (1,2)$,
\begin{align}
    x^{1-\alpha} - 1 = \frac{\sin(\alpha \pi)}{\pi} \int_0^1 f_t(x) t^{\alpha - 1}(1-t)^{1-\alpha} \md t.\label{eq: power function integral}
\end{align}
The integral on the right-hand side can be approximated via GRJ quadratures by~\cite[Theorem 1]{faust2023rational} or~\cite[Proposition 3.12]{hahn2024bounds},
\begin{align}
        \sum_{j=1}^m w_j^1 f_{t_j^1}(x) \leq \int_0^1 f_t(x) & t^{\alpha - 1}(1-t)^{1-\alpha} \md t \leq \sum_{j=1}^m w_j^0 f_{ t_j^0}(x).
        \label{eq:quadrature_Petz_a}
    \end{align}
Here, $\{t_j^0\}_{j=1}^m$ and $\{w_j^0\}_{j=1}^m$ denote the nodes and weights of the GRJ quadrature with a fixed node at $t_1^0 = 0$, while $\{t_j^1\}_{j=1}^m$ and $\{w_j^1\}_{j=1}^m$ correspond to the nodes and weights with a fixed node at $t_m^1 = 1$. As $m \to \infty$, both the lower and upper bounds converge to the exact value of the integral. Thus, the integral in Eq.~\eqref{eq:quadrature_Petz_a} can be approximated using either bound. For simplicity, we omit the superscripts in the following discussion.

Combining Eqs.~\eqref{eq: power function integral} and~\eqref{eq:quadrature_Petz_a}, we have
\begin{align}
    x^{1-\alpha} &\approx h_m(x) := 1 + \frac{\sin(\alpha \pi)}{\pi} \sum_{j=1}^m w_j f_{t_j}(x).
\end{align}
The performance of this quadrature approximation is shown in Figure~\ref{fig:quadrature_approx} (b), where we take $\alpha = 1.5$ as an example. As given in Eq.~\eqref{eq: power function f divergence}, we have the approximation to $Q_\alpha(\rho\|\sigma)$ by
\begin{align}
    &Q_\alpha(\rho\|\sigma) \approx D_{h_m}(\rho\|\sigma) = 1 + \frac{\sin(\alpha \pi)}{\pi}\sum_{j=1}^m w_j D_{f_{t_j}}(\rho\|\sigma),
    \label{eq:petz_alpha_vari}
\end{align}
where each $D_{f_{t_j}}(\rho\|\sigma)$ on the right-hand side can be evaluated using the variational expression provided in Lemma~\ref{lemma:variational_ft}.

In the case of $\alpha = 2$, we can consider $t=1$ and 
the standard quantum $f$-divergence gives~\cite[Remark 3.6]{brown2024device} the following relation,
\begin{align}
    &= 1 - D_{f_1}(\rho\|\sigma) ~~\text{and}~~ D_2(\rho\|\sigma) = \log Q_2(\rho\|\sigma).
    \label{eq:approx_tr_petz2}
\end{align}

\section{Parameterization of Hermitian polynomial operators}
\label{sec:polynomial_Z}

The variational expression introduced in the previous section is a special case of the following general optimization problem:
\begin{align}
    \inf_{\{Z_1,\dots,Z_N\}} \left\{\tr[\rho\, P(Z_1, \dots, Z_N)] + \tr[\sigma\, Q(Z_1, \dots, Z_N)]\right\},
    \label{eq:poly_Z}
\end{align}
where $\{Z_1,Z_2,\dots,Z_N\}$ is a set of linear operators, and $P$ and $Q$ are Hermitian polynomials satisfying $P^\dagger = P$ and $Q^\dagger = Q$. In this section, we present a method for parameterizing the polynomial observables $P$ and $Q$, along with a sampling procedure for estimating $\tr[\rho P]$ and $\tr[\sigma Q]$ that is well-suited for implementation on quantum computers.

Our method is based on the singular value decomposition and the extended SWAP test~\cite{kieferova2021quantum}. Without loss of generality, we consider the $p$-th order term in the polynomial, which takes the form
\begin{align}
    \sum_{\bm{i}\in [N]^p}a_{\bm{i}}\left(\cZ_{\bm{i}} + \mathcal{Z}_{\bm{i}}^\dagger \right),~~\text{where}~~a_{\bm{i}}\in \mathbb{R}, ~\cZ_{\bm{i}} = \prod_{k=1}^p Z_{i_k},
    \label{eq:def_cZ_gamma}
\end{align}
where $\bm{i} := (i_1, i_2, \dots,i_p)\in [N]^p$ is an index-sequence of length $p$. For possible terms such as $Z_1Z_1^\dagger$, it can be included by considering a polynomial $P(Z_1,Z_1^\dagger,\dots)$. To parameterize the polynomial $P(Z_1, \dots,Z_N)$ and estimate the value $\tr[\rho P]$, it suffices to parameterize the linear operator $\cZ_{\bm{i}}$ and then estimate the value $\tr[\rho (\cZ_{\bm{i}}+\cZ_{\bm{i}}^\dagger)]$. Detailed procedures are given as follows.

\vspace{0.3cm}
\noindent
\textbf{Parameterization. }With singular value decomposition, a linear operator can be parameterized by
\begin{align}
    Z_{i_k} = \sum_{j_k = 1}^d \lambda_{k,j_k} U_{i_k}\ket{j_k}\bra{j_k}V_{i_k},
\end{align}
where $\lambda_{k,j_k}\ge 0$ and $\{\ket{1},\ket{2},\dots,\ket{d}\}$ denotes a set of computational basis. Then we have
\begin{align}
    \cZ_{\bm{i}}
    &= \sum_{(j_1,\dots, j_p)\in[d]^p}\prod_{k=1}^p \lambda_{k,j_k} U_{i_k}\ket{j_k}\bra{j_k}V_{i_k}.
    \label{eq:para_poly_Z}
\end{align}
We parameterize the unitaries $U_{i_k}$ and $V_{i_k}$ with a set of parameterized quantum circuits $U_{i_k}(\bm{\theta}_{i_k})$ and $V_{i_k}(\bm{\beta}_{i_k})$, where the circuit structures are tailored to the specific problems being addressed and the available quantum hardware. The dimensions of vectors $\bm{\theta}_{i_k}$ and $\bm{\beta}_{i_k}$ depend on the structure of the parameterized quantum circuits.

When the linear operators $Z_i$ are bounded, Eq.~\eqref{eq:poly_Z} transforms into a constrained optimization problem, expressed as $0\le \|Z_i\|\le M_i$ for all $i$, where $\|\cdot\|$ denotes a suitable operator norm and $M_i$ are given constants. This optimization can be efficiently handled by imposing constraints on the parameters $\lambda_{k,j_k}$ during data processing on classical computers. Moreover, an alternative approach for parameterizing a linear operator, which involves incorporating classical neural networks (CNNs) as used in~\cite{goldfeld2024quantum}, is detailed in Appendix~\ref{sec:CNN}.

\vspace{0.3cm}
\noindent
\textbf{Sampling.} Denote $\bm{j} := (j_1, j_2 \dots, j_p)\in [d]^p$, $a_{\bm{j}} := \prod_{k=1}^p \lambda_{k,j_k}$, and
\begin{align}
    \Theta_{\bm{j},\bm{i}} := \mathrm{Re} \left[\tr\left(\rho\prod_{k=1}^p U_{i_k}\ket{j_k}\bra{j_k}V_{i_k}\right) \right].
    \label{eq:remark2_need1}
\end{align}
Due to Eq.~\eqref{eq:para_poly_Z}, after parameterization, the quantity we need to estimate becomes
\begin{align}
\tr[\rho(\cZ_{\bm{i}} + \cZ_{\bm{i}}^\dagger)] &= 2\mathrm{Re}\tr[\rho \cZ_{\bm{i}}] = 2\sum_{\bm{j} \in [d]^p} a_{\bm{j}}\Theta_{\bm{j},\bm{i}}.
\end{align}
We can estimate the quantity above based on the extended SWAP test provided in~\cite[Theorem 2.3]{kieferova2021quantum}, as adapted in the following lemma.

\begin{lemma}
\label{the:extended_SWAP}
Let $\rho$ be a quantum state of dimension $d$, and $\cZ_{\bm{i}}$ be the linear operator defined in Eq.~\eqref{eq:def_cZ_gamma}, the corresponding parameterization of which is Eq.~\eqref{eq:para_poly_Z}. Denote $\bm{j} = (j_1,j_2,\dots,j_p)\in [d]^p$ and $\ket{\bm{j}} = \ket{j_1}\ket{j_2}\cdots \ket{j_p}$. Then the unitary circuit $\chi_{\bm{i}}$ shown in Figure~\ref{fig:general_chi} gives
\begin{align}
    \Theta_{\bm{j},\bm{i}} = 2\tr\left[\left(\ket{0}\bra{0}\ox I\right)\chi_{\bm{i}}\left(\ket{0}\bra{0}\ox \rho\ox \ket{\bm{j}}\bra{\bm{j}}\right)\chi_{\bm{i}}^\dagger\right] - 1. 
\end{align}
\end{lemma}

More specifically, the unitary operation $\chi_{\bm{i}}$ consists of two Hadamard gates, a series of controlled unitaries, and a controlled cyclic permutation. We implement the unitary $\chi_{\bm{i}}$ on the input state $\ket{0}\bra{0}\ox \rho\ox \ket{\bm{j}}\bra{\bm{j}}$, and then measure the first qubit. The probability of outcome $0$ of the measurement is given by $(\Theta_{\bm{j},\bm{i}}+1)/2$. 

\begin{figure}[ht]
    \centering
    \includegraphics[width=\linewidth]{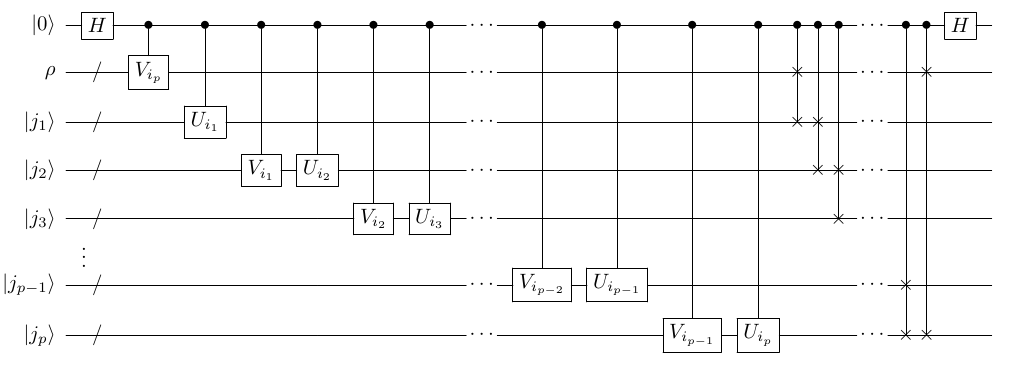}
    \caption{Visualization of unitary operation $\chi_{\bm{i}}$. We evolve the input state $\ket{0}\bra{0}\ox \rho\ox \ket{\bm{j}}\bra{\bm{j}}$ with a unitary operation $\chi_{\bm{i}}$, and then measure the first qubit in the computational basis. The probability of outcome $0$ of the measurement gives the value of $(\Theta_{\bm{j},\bm{i}}+1)/2$. 
    }
    \label{fig:general_chi}
\end{figure}

With the parameterization and sampling procedure mentioned above, we can estimate the loss function in the optimization problem Eq.~\eqref{eq:poly_Z} by sampling on quantum computers, which paves the way for estimating quantum relative entropies (where $p=1$) in the subsequent sections.

\section{Variational quantum algorithms}
\label{sec:VQAs}

In this section, we apply the general method introduced in Section~\ref{sec:polynomial_Z} to estimate the standard $f_t$-divergence via variational quantum algorithms. By employing this as a subroutine, we can then estimate quantum relative entropies by utilizing the approximation from Eqs.~\eqref{eq:approx_rel_entr} and~\eqref{eq:petz_alpha_vari}.

\subsection{Standard quantum $f$-divergence}

Based on the variational expression in Eq.~\eqref{eq:variant_Df}, the evaluation procedure involves parameterizing the linear operator $Z$, estimating the individual terms $\tr[\rho(Z+Z^\dagger)]$, $\tr[\rho Z^\dagger Z]$, and $\tr[\sigma ZZ^\dagger]$, and finally optimizing the loss function. The detailed steps are as follows.

\vspace{0.3cm}
\noindent
\textbf{Parameterization.} The first step is to parameterize the linear operator $Z$ with singular value decomposition, i.e., $Z = U\Lambda V$, where $U$ and $V$ are two unitary operators, $\Lambda$ is non-negative real diagonal matrix. We choose a set of basis $\{\ket{i}\}_{i=1}^d$ and parameterize $\Lambda$ as $\sum_{i=1}^d \lambda_i\ket{i}\bra{i}$, where $\lambda_i \ge 0$ and $\lambda_i \in \mathbb{R}$ for all $i$. All parameters $\lambda_i$ are collectively represented as $\boldsymbol{\lambda}$. Then the unitaries $U$ and $V$ are represented by two parameterized quantum circuit $U(\bm{\theta})$ and $V(\bm{\beta})$ with parameters $\bm{\theta}\in \mathbb{R}^q$ and $\bm{\beta}\in \mathbb{R}^r$. For the sake of simplicity, we denote $U(\bm{\theta})$ as $U_{\bm{\theta}}$ and $V(\bm{\beta})$ as $V_{\bm{\beta}}$. Then the linear operator $Z$ is then parameterized as
\begin{align}
    Z = \sum_{i = 1}^d \lambda_i U_{\bm{\theta}} \ket{i}\bra{i}V_{\bm{\beta}}.
\end{align}

Note that the contribution of each term in Eq.~\eqref{eq:variant_Df} depends on how the linear operator $Z$ acts within the support of $\rho$ and $\sigma$. In practice, when the quantum states are not full-rank, i.e., $\mathrm{rank}(\rho) = \mathrm{rank}(\sigma) = s$, we can restrict $Z$ to an at most $s$-dimensional subspace without affecting the value of the optimization problem. Any contribution from $Z$ outside the support of $\rho$ and $\sigma$ does not affect the traces involving $\rho$ and $\sigma$. Thus, to ensure the efficiency, we can parameterize $\Lambda$ with a sparse diagonal matrix that only contains at most $s$ nonzero elements, i.e., $\Lambda = \sum_{k=1}^s \lambda_{i_k} \ket{i_k}\bra{i_k}$, where $i_k \in [d]$.

\vspace{0.2cm}
\noindent
\textbf{Sampling.} With the above parameterization, we obtain
\begin{align}
\tr[\rho Z^\dagger Z] & = \sum_{i=1}^d \lambda_i^2 \tr\left[\rho V_{\bm{\beta}}^\dagger \ket{i}\bra{i} V_{\bm{\beta}}\right], \\
\tr[\sigma ZZ^\dagger] & = \sum_{i=1}^d \lambda_i^2 \tr\left[\sigma U_{\bm{\theta}}\ket{i}\bra{i}U_{\bm{\theta}}^\dagger\right],  \\
\tr[\rho(Z + Z^\dagger)] & = 4 \sum_{i=1}^d \lambda_i\tr\left[(\ket{0}\bra{0}\ox I)\chi (\ket{0}\bra{0}\ox \rho \ox \ket{i}\bra{i})\chi^\dagger\right] -2,
\end{align}
where $\chi$ is a unitary operation constructed according to the quantum circuit shown in Figure~\ref{fig:unitary_chi}.
Then each individual term above can be estimated on quantum computers. 
\begin{figure}[H]
    \centering
    \includegraphics[width=0.45\linewidth]{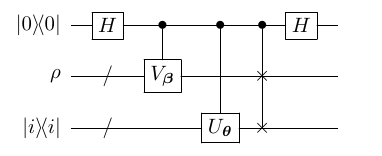}
    \caption{Visualization of unitary operation $\chi$. It consists of two Hadamard gates $H$, controlled $U_{\bm{\theta}}$, $V_{\bm{\beta}}$ gates, and a controlled swap gates.}
    \label{fig:unitary_chi}
\end{figure}
More specifically, denote the probability distributions as
\begin{subequations}
    \begin{align} 
        &p_{\bm{\beta}}^{(i)} := \tr\left[V_{\bm{\beta}} \rho V_{\bm{\beta}}^\dagger \ket{i}\bra{i} \right], \\
        &p_{\bm{\theta}}^{(i)} := \tr\left[U^\dagger_{\bm{\theta}}\sigma U_{\bm{\theta}}\ket{i}\bra{i}\right],\\ 
        &p_{\chi}^{(i)} := \tr\left[(\ket{0}\bra{0}\ox I)\chi (\ket{0}\bra{0}\ox \rho \ox \ket{i}\bra{i})\chi^\dagger\right].
    \end{align}
    \label{eq:sample_p}
\end{subequations}
Then, \( p_{\bm{\beta}}^{(i)} \) can be estimated by evolving the quantum state \( \rho \) using \( V_{\bm{\beta}} \) and performing measurements in the computational basis. The same procedure applies to \( p_{\bm{\theta}}^{(i)} \). Finally, \( p_{\chi}^{(i)} \) can be estimated by evolving the initial state \( \ket{0}\bra{0} \otimes \rho \otimes \ket{i}\bra{i} \) with the unitary \( \chi \) and measuring the first qubit in the computational basis. With these sampling results, we can obtain from classical computations that 
\begin{align}
\tr[\rho Z^\dagger Z] & = \sum_{i=1}^d \lambda_i^2 p_{\bm{\beta}}^{(i)}, \\
\tr[\sigma ZZ^\dagger] & = \sum_{i=1}^d \lambda_i^2 p_{\bm{\theta}}^{(i)}, \\
\tr[\rho(Z + Z^\dagger)] & = 4 \sum_{i=1}^d \lambda_i p_{\chi}^{(i)} -2.
\end{align}

We are now well-prepared to introduce our variational quantum algorithms.

\vspace{0.3cm}
\noindent
\textbf{Algorithms.} The idea is to solve the optimization problem defined in Eq.~\eqref{eq:variant_Df} for $m$ times to get the estimates of $D_{f_{t_j}}(\rho\|\sigma)$ for $j \in \{1,\cdots,m\}$, and then obtain the estimation value of $D_{-r_m}$ with Eq.~\eqref{eq:approx_rel_entr} and $D_{h_m}(\rho\|\sigma)$ with Eq.~\eqref{eq:petz_alpha_vari}, which approximate $D(\rho\|\sigma)$ and $Q_\alpha(\rho\|\sigma)$, respectively. In the case of $\alpha = 2$, the approximation of $Q_2(\rho\|\sigma)$ is directly obtained by the estimation of $D_{f_1}(\rho\|\sigma)$. When estimating $D_{f_{t}}$, the loss function is given as follows:
\begin{align}
    \cL(\bm{\lambda}, \bm{\theta},\bm{\beta}) := &\sum_{i=1}^d \left\{ t\lambda_i^2 p_{\bm{\theta}}^{(i)} + (1-t)\lambda_i^2 p_{\bm{\beta}}^{(i)} +\lambda_i(4p_{\chi}^{(i)} - 2) \right\},
    \label{eq:loss_vqa_rel}
\end{align}
which corresponds exactly to the function we aim to minimize in in Eq.~\eqref{eq:variant_Df}. The optimization problem is thus $\hat{\cL} = \min _{\bm{\lambda},\bm{\theta},\bm{\beta}} \mathcal{L}(\bm{\lambda},\bm{\theta},\bm{\beta})$, and the approximation value of $D_{f_t}(\rho\|\sigma)$ is $\hat{D}_{f_t}(\rho\|\sigma) = \frac{1}{t}(1+\hat{\cL})$. Since the loss function is quadratic with respect to each $\lambda_i$, we can write the solution of $\lambda_i$ as follows:
\begin{align}
    \lambda_i^{\text{opt}} = \max\left\{ 0, \frac{1- 2p_\chi^{(i)}}{t p_\theta^{(i)} + (1-t)p_\beta^{(i)}} \right\},
    \label{eq:lambda_opt}
\end{align}
and the optimization problem is simplified as
\begin{align}
    \hat{\cL} = \min_{\bm{\theta},\bm{\beta}} \mathcal{L}(\bm{\lambda}^{\text{opt}},\bm{\theta},\bm{\beta}).
    \label{eq:def_optimization}
\end{align}
We utilize the gradient descent algorithm to optimize the loss function defined by Eq.~\eqref{eq:loss_vqa_rel}, which requires the gradients of the loss function with respect to each parameter. The values of $p_{\bm{\theta}}^{(i)}$, $p_{\bm{\beta}}^{(i)}$, and $p_{\bm{\chi}}^{(i)}$ are estimated by sampling on quantum computers. The quantum circuit parameters $\bm{\theta}$ and $\bm{\beta}$ are updated with respect to the parameter-shift rule~\cite{schuld2019evaluating}. The pseudocode of the algorithm for estimating $D_{f_t}$ is shown in Algorithm~\ref{alg:algorithm_D_ft}. 

\begin{algorithm}[ht]
\caption{Variational quantum algorithm for estimating standard $f_t$-divergence}\label{alg:algorithm_D_ft}
\LinesNumbered
\BlankLine

\KwIn{\\  \begin{tabular}{p{2.5cm}l}
     $K$ & number of iterations \\
     $\ell_r$ & learning rate\\
     $N$ & number of samples\\
     $t$ & function parameter \\
     $\rho$, $\sigma$ & quantum states with the same support
\end{tabular}} 

\BlankLine

\KwOut{\\ \begin{tabular}{p{2.5cm}l} 
$\hat D_{f_t}(\rho\|\sigma)$ & estimate of $D_{f_t}(\rho\|\sigma)$
\end{tabular}}

\BlankLine

$\bm{\theta}^1\leftarrow $ Random initialization in $[0,2\pi]^q$.

$\bm{\beta}^1\leftarrow $ Random initialization in $[0,2\pi]^r$.

\If{$t=0$}{

\Return $\hat{D}_{f_0}(\rho\|\sigma) \leftarrow 0$.
}
\ForEach{$k \in \{1, 2, \ldots, K\}$}{

Sample $N$ times to evaluate $p_{\bm{\theta}}^{(i)}$, $p_{\bm{\beta}}^{(i)}$ and $p_{\chi}^{(i)}$ for $i = 1,\dots,d$, with $\bm{\theta}^k$ and $\bm{\beta}^k$ on quantum computers.

Compute $\bm{\lambda}^{\text{opt}}$ with Eq.~\eqref{eq:lambda_opt}.

Sample $N$ times to evaluate $\nabla_{\bm{\theta}} \cL$ and $\nabla_{\bm{\beta}} \cL$ using the parameter shift rule.

$\bm{\theta}^{k+1}\leftarrow \bm{\theta}^k + \ell_r\nabla_{\bm{\theta}} \cL$ \\$\bm{\beta}^{k+1}\leftarrow \bm{\beta}^k + \ell_r\nabla_{\bm{\beta}} \cL$
}

$\hat{\mathcal{L}} \leftarrow \mathcal{L}({\bm{\theta}}^{K+1}, {\bm{\beta}}^{K+1})$.

\Return $\hat{D}_{f_t}(\rho\|\sigma) \leftarrow (1+\hat \cL)/t$
\end{algorithm}

It is worth noting that the only step requiring quantum computers is obtaining the sampling results $p_{\bm{\theta}}^{(i)}$, $p_{\bm{\beta}}^{(i)}$ and $p_{\chi}^{(i)}$. All other computations can be performed on classical computers.  Furthermore, the quantum sampling procedure is naturally suited to distributed scenarios. Specifically, according to Eq.~\eqref{eq:sample_p}, the probability distributions \(p_{\bm{\theta}}^{(i)}\) and \(p_\chi^{(i)}\) can be sampled on a quantum computer preparing the state \(\rho\), while \(p_{\bm{\beta}}^{(i)}\) can be obtained on a quantum computer preparing \(\sigma\). This enables the comparison of quantum data generated by distributed quantum devices, making the approach directly applicable to cross-platform verification and related distributed quantum information tasks~\cite{elben2020cross,greganti2021cross,knorzer2023cross,zheng2024cross}.

\subsection{Quantum relative entropy}
Recall that the approximation of quantum relative entropy is given by
\begin{align}
    \hat D(\rho\|\sigma) \approx \hat{D}_{-r_m}(\rho\|\sigma) = -\sum_{j=1}^m \frac{w_j }{\ln 2} \hat{D}_{f_{t_j}}(\rho\|\sigma),
\end{align}
where $\{t_j\}_{j=1}^m$ is the set of GRJ quadrature nodes with fixed node $t_1 = 0$ or $t_m = 1$. Based on Algorithm~\ref{alg:algorithm_D_ft}, the algorithm for evaluating quantum relative entropy is shown in Algorithm~\ref{alg:algorithm_rel}.

\vspace{0.2cm}
\begin{algorithm}[H]
\caption{Variational quantum algorithm for estimating quantum relative entropy}\label{alg:algorithm_rel}
\LinesNumbered
\BlankLine

\KwIn{\\  \begin{tabular}{p{2.5cm}l}
     $K$ & number of iterations \\
     $\ell_r$ & learning rate\\
     $N$ & number of samples\\
     $m$ & number of quadrature nodes \\
     $\{t_j,w_j\}_{j=1}^m$ & GRJ quadrature nodes and weights \\
     $\rho$, $\sigma$ & quantum states with the same support
\end{tabular}} 

\BlankLine

\KwOut{\\ \begin{tabular}{p{2.5cm}l} 
$\hat D(\rho\|\sigma)$ & estimate of $D(\rho\|\sigma)$
\end{tabular}}

\BlankLine

\ForEach{$j \in \{1, 2, \ldots, m\}$}{
Run Algorithm~\ref{alg:algorithm_D_ft} with inputs $K$, $\ell_r$, $N$, and $t= t_j$. 

Save $\hat{D}_{f_{t_j}}(\rho\|\sigma)$.
}

$\hat{D}_{-r_m}(\rho\|\sigma) \leftarrow -\sum_{j=1}^m \frac{w_j}{\ln 2}\hat{D}_{f_{t_j}}(\rho\|\sigma)$.

\Return $\hat{D}(\rho\|\sigma) \leftarrow \hat{D}_{-r_m}(\rho\|\sigma)$
\end{algorithm}

\subsection{Petz \Renyi divergence}
\label{sec:alg_Petz_a}

Recall that the approximation of $Q_\alpha(\rho\|\sigma)$ is given by
\begin{align}
    \hat{Q}_\alpha(\rho\|\sigma) \approx \hat D_{h_m}(\rho\|\sigma) = 1 + \frac{\sin(\alpha \pi)}{\pi}\sum_{j=1}^m w_j \hat{D}_{f_{t_j}}(\rho\|\sigma), 
\end{align}
where $\{t_j\}_{j=1}^m$ is the set of GRJ quadrature nodes with fixed node $t_1=0 $ or $t_m = 1$. We note here that the quadrature nodes for estimating $Q_\alpha(\rho\|\sigma)$ are different from those for estimating $D_{-r_m}(\rho\|\sigma)$, since they depend on the parameter $\alpha$. The approximation of $D_\alpha(\rho\|\sigma)$ is thus 
\begin{align}
    \hat{D}_\alpha(\rho\|\sigma)=\frac{1}{\alpha - 1}\log \hat{Q}_\alpha(\rho\|\sigma).
\end{align}
The algorithm for estimating the Petz \Renyi divergence is shown in Algorithm~\ref{alg:algorithm_Petz_a}. In the case of $\alpha = 2$, the estimate of $Q_2(\rho\|\sigma)$ is obtained by
\begin{align}
    \hat{Q}_2(\rho\|\sigma) = 1-\hat{D}_{f_1}(\rho\|\sigma),
\end{align}
and the approximation of $D_2(\rho\|\sigma)$ is thus $\hat{D}_2(\rho\|\sigma) = \log \hat{Q}_2(\rho\|\sigma)$, which can be obtained by running Algorithm~\ref{alg:algorithm_D_ft} with input $t = 1$.

\begin{algorithm}[H]
\caption{Variational quantum algorithm for estimating Petz \Renyi divergence} \label{alg:algorithm_Petz_a}
\LinesNumbered
\BlankLine

\KwIn{\\  \begin{tabular}{p{2.5cm}l}
     $K$ & number of iterations \\
     $\ell_r$ & learning rate\\
     $N$ & number of samples\\
     $m$ & number of quadrature nodes \\
     $\{t_j,w_j\}_{j=1}^m$ & GRJ quadrature nodes and weights \\
     $\rho$, $\sigma$ & quantum states with the same support
\end{tabular}} 

\BlankLine

\KwOut{\\ \begin{tabular}{p{2.5cm}l} 
$\hat D_\alpha(\rho\|\sigma)$ & estimate of $D_\alpha(\rho\|\sigma)$
\end{tabular}}

\BlankLine

\ForEach{$j \in \{1, 2, \ldots, m\}$}{
Run Algorithm~\ref{alg:algorithm_D_ft} with inputs $K$, $\ell_r$, $N$, and $t= t_j$. 

Save $\hat{D}_{f_{t_j}}(\rho\|\sigma)$.
}

$\hat{Q}_\alpha(\rho\|\sigma) \leftarrow 1 + \frac{\sin(\alpha \pi)}{\pi}\sum_{j=1}^m w_j \hat{D}_{f_{t_j}}(\rho\|\sigma)$.

\Return $\hat D_\alpha(\rho\|\sigma) \leftarrow \frac{1}{\alpha - 1}\log\hat{Q}_\alpha(\rho\|\sigma)$
\end{algorithm}

\section{Convergence analysis}
\label{sec:converge_analysis}

In this section, we prove that our loss function is operator-convex, ensuring a well-behaved optimization
landscape. We further discuss the convergence and trainability of our approach, and present numerical evidence indicating that our loss function is absent of the barren plateau phenomenon.

\subsection{Convexity and global optimality}

Let $\mathfrak{L}(\mathcal{H})$ denote the set of all bounded linear operators on the Hilbert space $\mathcal{H}$. Consider a function $g:\mathfrak{L}(\mathcal{H}) \to \mathbb{R}$. We say that $g$ is operator-convex~\cite{boyd2004convex} if, for all $X, Y \in \mathfrak{L}(\mathcal{H})$ and $\omega \in [0,1]$, the following holds: $g((1-\omega)X + \omega Y) \leq (1-\omega)g(X) + \omega g(Y)$. 

We show that the loss function in Eq.~\eqref{eq:variant_Df} is operator-convex with respect to $Z$ and we also characterize its optimal solution.

\begin{lemma}
    Let $g(Z,Z^\dagger) = \tr[\rho] + \tr[\rho(Z+Z^\dagger)] + (1-t)\tr[\rho Z^\dagger Z] + t\tr[\sigma Z Z^\dagger]$, where $\supp(\rho)\subseteq \supp(\sigma)$. Then the function $g(Z,Z^\dagger)$ is operator-convex with respect to $Z$. Moreover, any solution $\tilde Z$ of the following Sylvester's equation,
    \begin{align}
        (1-t)\rho \tilde Z^\dagger + t\tilde Z^\dagger \sigma = -\rho ~~\text{ or }~~(1-t)\tilde Z \rho  + t \sigma \tilde Z = -\rho,
        \label{eq:Sylvester}
    \end{align}
    is an optimizer of the infimum problem $\inf_{Z\in \mathfrak{L}(\mathcal{H})}g(Z,Z^\dagger)$. The above equation always admits at least one solution, which is unique if and only if at least one of $\rho$ and $\sigma$ is full-rank.
    \label{lemma:Sylvester}
\end{lemma}
\begin{proof}
    First, we need to show $g(\omega Z_1 + (1-\omega)Z_2)\le \omega g(Z_1) + (1-\omega)g(Z_2)$ holds for any $Z_1$, $Z_2$ and $\omega\in[0,1]$. Let $\ket{\mathrm{vec}(A)}:= \sum_{ij} \<i|A|j\> \ket{i}\ket{j}$ be the vectorization of matrix $A$. According to the vectorization equality $\ket{\mathrm{vec}(AXB)} = A\ox B^\mathrm{T} \ket{\mathrm{vec}(X)}$, we have $g(Z) = \tr[\rho] + 2\mathrm{Re}\braket{\mathrm{vec}(\rho)}{\mathrm{vec}(Z)} + \bra{\mathrm{vec}(Z)}A_t\ket{\mathrm{vec}(Z)}$, where $A_t = (1-t) I\ox \rho^* + t\sigma \ox I$. Thus, 
    \begin{align}
        & \omega g(Z_1) + (1-\omega)g(Z_2) - g(\omega Z_1 + (1-\omega)Z_2) \notag\\
        =& \omega  \bra{\mathrm{vec}(Z_1)}A_t\ket{\mathrm{vec}(Z_1)} + (1-\omega)\bra{\mathrm{vec}(Z_2)}A_t\ket{\mathrm{vec}(Z_2)}- \omega^2 \bra{\mathrm{vec}(Z_1)}A_t\ket{\mathrm{vec}(Z_1)} \\
        &- (1-\omega)^2 \bra{\mathrm{vec}(Z_2)}A_t\ket{\mathrm{vec}(Z_2)} - 2\omega(1-\omega) \mathrm{Re} \bra{\mathrm{vec}(Z_1)}A_t\ket{\mathrm{vec}(Z_2)}\\
        \ge &0,
    \end{align}
    which proves the convexity, where the last inequality follows from $2\mathrm{Re}\bra{\mathrm{vec}(Z_1)}A_t\ket{\mathrm{vec}(Z_2)}\le \bra{\mathrm{vec}(Z_1)}A_t\ket{\mathrm{vec}(Z_1)} + \bra{\mathrm{vec}(Z_2)}A_t\ket{\mathrm{vec}(Z_2)}$. 

    By the convexity of the loss function, the solution of the infimum problem $\inf_{Z\in \mathfrak{L}(\mathcal{H})}g(Z,Z^\dagger)$ should satisfy $\frac{\partial g}{\partial Z} = 0$ and $\frac{\partial g}{\partial Z^\dagger} = 0$. Taking the variation of $g$ with respect to $Z$ yields $\delta g(Z,Z^\dagger) = \tr[\left(\rho + (1-t) \rho Z^\dagger + tZ^\dagger \sigma \right)\delta Z] = 0$, which leads to the Sylvester's equation Eq.~\eqref{eq:Sylvester}. Due to the Sylvester's theorem~\cite[Theorem 2.4.4.1]{horn2012matrix}, Eq.~\eqref{eq:Sylvester} is uniquely solvable if and only if $(1-t)\rho$ and $-t\sigma$ do not share any eigenvalue. Since $t\in(0,1]$, $\rho$ and $\sigma$ are positive semidefinite and $\supp(\rho)\subseteq \supp(\sigma)$, we have that $(1-t)\rho$ and $-t\sigma$ can share the eigenvalue $0$ only if both are not full-rank. It remains to prove the existence of the solution of the Sylvester's equation when both $\rho$ and $\sigma$ are not full-rank. In this case, $\rho$ and $\sigma$ have zero eigenvalues. A necessary and sufficient condition of the existence of the solution is given by $P_\rho \rho P_\sigma = 0$~\cite{horn2012matrix}, where $P_\rho$ and $P_\sigma$ denote the orthogonal projections onto their kernels $\ker \rho$ and $\ker \sigma$. Since $P_\rho \rho = 0$, $P_\rho \rho P_\sigma = 0$, and thus Eq.~\eqref{eq:Sylvester} has at least one solution, which completes the proof.
\end{proof}

Lemma~\ref{lemma:Sylvester} shows that our loss function is operator-convex in $Z$, implying that any local minimum is also a global minimum. Consequently, our optimization landscape is free of spurious local minima, ensuring a simple and well-behaved structure. 

\subsection{Convergence, trainability and barren plateau}

The convergence of a variational quantum algorithm is fundamentally determined by the structure of its loss function and the resulting optimization landscape. As shown in the previous section, our loss function is operator-convex, guaranteeing the existence of a global optimum at the operator level. This means that, provided the parameterized circuit is sufficiently expressive to represent all linear operators, the optimal solution can be reached. Achieving adequate expressibility typically requires deeper or more complex circuit architectures, which may increase computational overhead and resource requirements. In practice, this can be addressed by employing a well-designed ansatz and repeating it over a moderate number of layers. As the circuit depth increases, the expressibility of the ansatz improves, enhancing the chances of finding a better solution.

However, prior studies have shown that highly expressive parameterized quantum circuits can suffer from barren plateaus—regions in the optimization landscape where gradients vanish exponentially, making training infeasible~\cite{mcclean2018barren,cerezo2021variational}. In particular, Holmes et al.~\cite{holmes2022connecting} demonstrated that for a broad class of loss functions, such as energy expectation values $C(\theta) = \tr[H U(\theta)\rho U^\dagger(\theta)]$, there is a fundamental trade-off: highly expressive ansatzes tend to produce vanishing gradients, resulting in barren plateaus. Nevertheless, for certain specialized cost functions—particularly those with specific symmetries or locality properties—this trade-off does not always apply, even for highly expressive circuits. This raises a key question: can our loss function avoid the barren plateau phenomenon?

Intuitively, our loss function defined in Eq.~\eqref{eq:loss_vqa_rel} may partially circumvent the barren plateau problem. Unlike most cost functions, our loss function combines contributions from terms such as $\mathrm{Tr}[\rho Z^{\dagger}Z]$, $\mathrm{Tr}[\sigma ZZ^{\dagger}]$, and $\mathrm{Tr}[\rho(Z+Z^{\dagger})]$, which is non-global as $Z^\dagger Z$, $ZZ^\dagger$ and $Z+Z^\dagger$ are decomposed into projectors $\ket{i}\bra{i}$ rather than a single global expectation value. Additionally, the mixture of quadratic and linear dependence on the eigenvalue weights $\lambda_i$ ensures that gradients are sourced from multiple, non-identical directions in parameter space, reducing the extent to which gradients are averaged out under random circuit exploration.

To verify this intuition, we conduct numerical experiments comparing two different loss functions. The first is our proposed loss function from Eq.~\eqref{eq:loss_vqa_rel}:
\begin{align}
    \begin{split}
        \ell_1(\bm{\theta},\bm{\beta}) = \sum_{i=1}^d &\left\{ t\lambda_i^2 \tr\left[U^\dagger_{\bm{\theta}}\sigma U_{\bm{\theta}}\ket{i}\bra{i}\right] + (1-t)\lambda_i^2 \tr\left[V_{\bm{\beta}} \rho V_{\bm{\beta}}^\dagger \ket{i}\bra{i} \right] \right.\\
    &\left.+\lambda_i(4\tr\left[(\ket{0}\bra{0}\ox I)\chi (\ket{0}\bra{0}\ox \rho \ox \ket{i}\bra{i})\chi^\dagger\right] - 2) \right\},
    \end{split}
\end{align}
where $\chi$ is the extended SWAP test circuit. The second loss function is a standard global cost:
\begin{align}
    \ell_2(\bm{\alpha}) = \tr[W_{\bm{\alpha}}\rho W^\dagger_{\bm{\alpha}}\sigma].
\end{align}

For all parameterized quantum circuits $U_{\bm{\theta}}$, $V_{\bm{\beta}}$, and $W_{\bm{\alpha}}$, we employ a fixed ansatz. Specifically, we adopt Circuit-$9$ from~\cite[Figure 2]{sim2019expressibility}, as illustrated in Figure~\ref{fig:ansatz_bp}, since it has been shown that its expressibility increases monotonically with the number of layers from $1$ to $5$.

\begin{figure}[H]
    \centering
    \includegraphics[width=0.35\linewidth]{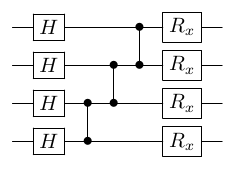}
    \caption{Ansatz for testing barren plateau phenomenon.}
    \label{fig:ansatz_bp}
\end{figure}

To assess the scale of the gradients, we compute $\partial_{\bm{\theta}}\ell_1(\bm{\theta},\bm{\beta})$, $\partial_{\bm{\beta}}\ell_1(\bm{\theta},\bm{\beta})$, and $\partial_{\bm{\alpha}}\ell_2(\bm{\alpha})$ at $N = 10{,}000$ uniformly random samples $\{\bm{\theta}_i,\bm{\beta}_i,\bm{\alpha}_i\}_{i=1}^{N}$, and evaluate the mean absolute values $\mathbf{E}[|\partial_{\bm{\theta}} \ell_1
(\bm{\theta},\bm{\beta})|]=\sum_{i=1}^N|\partial_{\bm{\theta}} \ell_1
(\bm{\theta}_i, \bm{\beta}_i)|/N$, $\mathbf{E}[|\partial_{\bm{\beta}} \ell_1
(\bm{\theta},\bm{\beta})|]=\sum_{i=1}^N|\partial_{\bm{\beta}} \ell_1(\bm{\theta}_i,\bm{\beta}_i)|/N$, and $\mathbf{E}[|\partial_{\bm{\alpha}} \ell_2
(\bm{\alpha})|] = \sum_{i=1}^N|\partial_{\bm{\alpha}} \ell_2(\bm{\alpha}_i)|/N$. For computational reasons, results for $5$-qubit cases are obtained using $N=1{,}000$ random samples.

The results are shown in Figure~\ref{fig:test_barren_plateau}. Figures~\ref{fig:grad_theta_layers}, \ref{fig:grad_beta_layers}, and \ref{fig:grad_alpha_layers} show the gradient scaling with respect to the number of layers. Notably, the gradient scaling of both $\ell_1(\bm{\theta},\bm{\beta})$ and $\ell_2(\bm{\alpha})$ does not decrease with the number of layers, despite the increasing expressibility of the parameterized circuits. However, as shown in Figure~\ref{fig:grad_5layers_qubit} (For visual clarity, the data vectors are normalized before plotting Figure~\ref{fig:grad_5layers_qubit}), the gradient scaling of $\ell_2(\bm{\alpha})$ decreases exponentially with respect to the number of qubits, while the gradient scaling of our loss function $\ell_1(\bm{\theta},\bm{\beta})$ remains essentially unchanged. 

These observations provide preliminary evidence that our proposed loss function circumvents the barren plateau phenomenon, maintaining trainability even as circuit expressibility increases. While the theoretical explanation for this behavior remains unclear, these results highlight the potential of our approach and motivate further systematic investigation in the future work.

\begin{figure}[H]
    \centering
    \begin{subfigure}{0.46\textwidth}
        \centering
        \includegraphics[width=\linewidth]{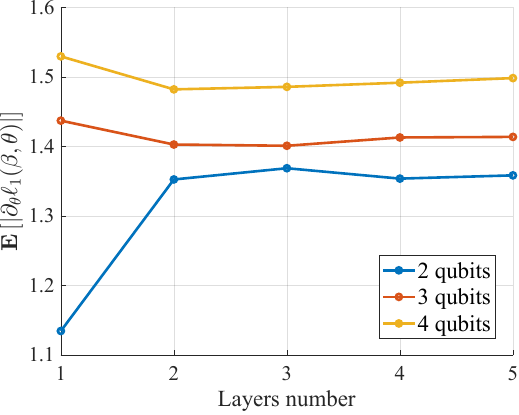}
        \caption{Gradient scaling of $\ell_1(\bm{\theta})$ vs. layers number}
        \label{fig:grad_theta_layers}
    \end{subfigure}
    \hfill
    \begin{subfigure}{0.46\textwidth}
        \centering
        \includegraphics[width=\linewidth]{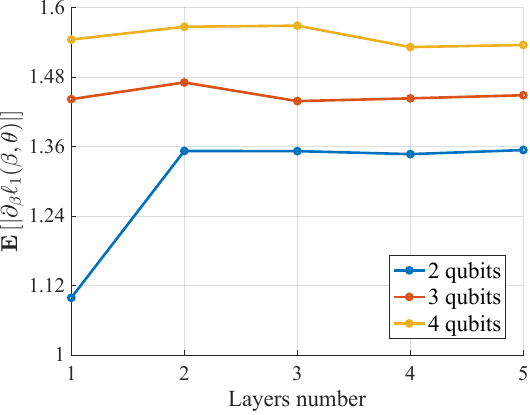}
        \caption{Gradient scaling of $\ell_1(\bm{\beta})$ vs. layers number}
        \label{fig:grad_beta_layers}
    \end{subfigure}

\vspace{1em}
    \begin{subfigure}{0.47\textwidth}
        \centering
        \includegraphics[width=\linewidth]{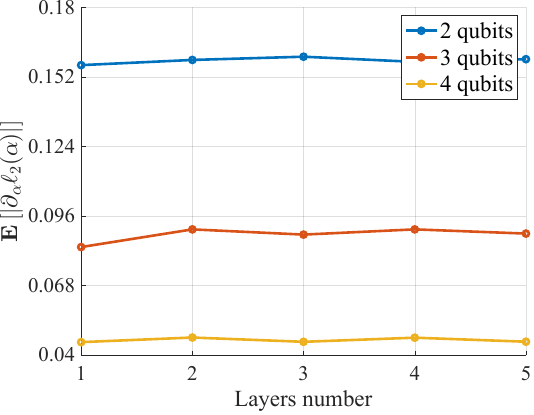}
        \caption{Gradient scaling of $\ell_2(\bm{\alpha})$ vs. layers number}
        \label{fig:grad_alpha_layers}
    \end{subfigure}
    \hfill
    \begin{subfigure}{0.46\textwidth}
        \centering
        \includegraphics[width=\linewidth]{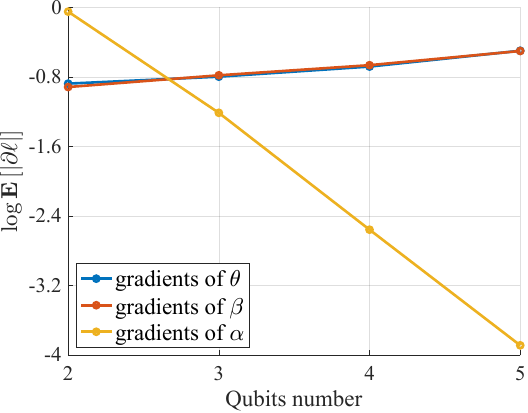}
        \caption{Gradient scaling vs. qubits number ($5$-layers)}
        \label{fig:grad_5layers_qubit}
    \end{subfigure}
    \caption{Gradient scaling with respect to number of layers and qubits.}
    \label{fig:test_barren_plateau}
\end{figure}

\section{Error analysis}
\label{sec:error_analysis}

In this section, we show that our estimates for quantum relative entropy and Petz \Renyi divergence serve as rigorous lower bounds for their true values, despite the heuristic nature of the algorithms. We also provide an error analysis for these estimates, given the relative error of Algorithm~\ref{alg:algorithm_D_ft}.

\subsection{Rigorous lower bounds}
\label{sec:lowerbound}

We first consider the estimate of quantum relative entropy. The rigorous lower bound is obtained using quadrature approximations with the fixed node at $t_1 = 0$. Denote $r^1_m(x)$ and $r^0_m(x)$ as the quadrature approximations of $\log x$ with fixed nodes at $t_m = 1$ and $t_1 = 0$, respectively. According to~\cite[Theorem 1]{faust2023rational}, we have $r^0_m(x) \ge \log x \ge r^1_m(x)$. The infimum optimization in Lemma~\ref{lemma:variational_ft} ensures that our approximation is always larger than the exact value, i.e., $\hat{D}_{f_t} (\rho\|\sigma) \ge D_{f_t}(\rho\|\sigma)$. Therefore, we have
\begin{align}
    D_{-r_m^0}(\rho\|\sigma) = -\sum_{j=1}^m \frac{w_j}{\ln2}D_{f_{t_j}}(\rho\|\sigma) \ge -\sum_{j=1}^m \frac{w_j}{\ln 2}\hat{D}_{f_{t_j}}(\rho\|\sigma) = \hat{D}(\rho\|\sigma).
\end{align}
This shows that our approximation serves as a lower bound, i.e.,
\begin{align}
    D(\rho\|\sigma) = D_{-\log}(\rho\|\sigma) \ge D_{-r_m^{0}}(\rho\|\sigma) \ge \hat{D}(\rho\|\sigma).
    \label{eq:exact_lowerbound}
\end{align}

The reasoning is analogous for Petz \Renyi divergence. Denote $h_m^1(x)$ and $h_m^0(x)$ as the quadrature approximations of $h(x) = x^{1-\alpha}$ with fixed nodes at $t_m = 1$ and $t_1 = 0$, respectively. For $\alpha \in (0,1)$, we have $h_m^1(x) \leq h(x) \leq h_m^0(x)$, while for $\alpha \in (1,2)$, we have $h_m^0(x) \leq h(x) \leq h_m^1(x)$. For $\alpha \in (0,1)$, it follows that $Q_\alpha(\rho\|\sigma) = D_h(\rho\|\sigma) \leq D_{h^0_m}(\rho\|\sigma) \leq \hat{D}_{h^0_m}(\rho\|\sigma)$. Similarly, for $\alpha \in (1,2)$, we have $Q_\alpha(\rho\|\sigma) = D_h(\rho\|\sigma) \geq D_{h_m^0}(\rho\|\sigma) \geq \hat{D}_{h_m^0}(\rho\|\sigma)$. Therefore, in both cases, the lower bound for Petz \Renyi divergences holds:
\begin{align}
    D_\alpha(\rho\|\sigma) = \frac{1}{\alpha - 1}\log Q_\alpha(\rho\|\sigma) \geq \frac{1}{\alpha - 1}\log \hat{D}_{h_m^0}(\rho\|\sigma) = \hat{D}_\alpha (\rho\|\sigma).
\end{align}

\subsection{Error bound for estimating quantum relative entropy}

The total error of our estimate comprises two components: $\varepsilon_Q$, which arises from the quadrature approximation and can be reduced by increasing the number of quadrature nodes $m$; and $\varepsilon_V$, which stems from the heuristic nature of the variational algorithms. Without loss of generality, we use the set of GRJ quadrature nodes with the fixed node $t_m = 1$ for this analysis. The error analysis for the algorithm using quadrature nodes with the fixed node $t_1 = 0$ is analogous and is detailed in Appendix~\ref{appendix:error_analysis}.

\begin{proposition}[Error bound for estimating $D(\rho\|\sigma)$]
\label{pro:error_rel_entr}
For any quantum states $\rho$ and $\sigma$ with their supports $\supp(\rho) \subseteq \supp(\sigma)$, if the relative error of Algorithm~\ref{alg:algorithm_D_ft} is within $\varepsilon_V$, i.e., 
\begin{align}
    \left|\frac{\hat{D}_{f_{t_j}}(\rho\|\sigma) - D_{f_{t_j}}(\rho\|\sigma)}{D_{f_{t_j}}(\rho\|\sigma)}\right| \le \varepsilon_V, ~~j = 1,2,\dots,m,
    \label{eq:epsilon_V_assum}
\end{align}
then the relative error of Algorithm~\ref{alg:algorithm_rel} can be bounded by
\begin{align}
    \left|\frac{\hat D(\rho\|\sigma) - D(\rho\|\sigma)}{D(\rho\|\sigma)}\right| \leq C\frac{1+\ve_V}{m^2} + \ve_V,
\end{align}
where $C := (Q_0(\rho\|\sigma) + Q_2(\rho\|\sigma))/(D(\rho\|\sigma)\ln2)$ and $m$ is the number of quadrature nodes.
\end{proposition}

\begin{proof}
We start by analyzing the difference between the quadrature approximation $r_m(x)$ and the exact function $\log x$. According to ~\cite[Proposition 2.2]{fawzi2022semidefinite}, we have the following inequality,
\begin{align}
    g(x) \leq -\log x \leq -r_m(x),
    \label{eq:bound_logx}
\end{align}
where $g(x) = -r_m(x) - \frac{1}{m^2\ln2}\left(\frac{1}{x}+x - 2\right)$. Note that $D_f(\rho\|\sigma)$ is linear and monotone in $f$ for fixed positive semidefinite operator $\rho$ and $\sigma$. This implies that $D_f(\rho\|\sigma) \le D_g(\rho\|\sigma)$ if $f(x)\le g(x)$ for all $x\in (0,+\infty)$. Due to Eq.~\eqref{eq:bound_logx}, we have $D_{g}(\rho\|\sigma) \leq D(\rho\|\sigma) \leq D_{-r_m}(\rho\|\sigma)$,
which indicates that the difference between $D(\rho\|\sigma)$ and $D_{-r_m}(\rho\|\sigma)$ must be smaller than the difference between $D_{-r_m}(\rho\|\sigma)$ and $D_g(\rho\|\sigma)$. Set $f_0(x) = x-1$ and $f_1(x) = (x-1)/x$. We get the following estimation
\begin{align}
    D_{-r_m}(\rho\|\sigma) - D(\rho\|\sigma) & \le  D_{-r_m}(\rho\|\sigma) - D_g(\rho\|\sigma) \\
    & = \frac{1}{m^2\ln2}D_{f_0 - f_1}(\rho\|\sigma) \\
    & = \frac{Q_0(\rho\|\sigma) + Q_2(\rho\|\sigma) - 2}{m^2\ln2} \\
    & \leq \frac{Q_0(\rho\|\sigma) + Q_2(\rho\|\sigma)}{m^2\ln2} 
    \label{eq:error_1}
\end{align}
where $D_{f_0}(\rho\|\sigma) = Q_0(\rho\|\sigma) - 1$ and $D_{f_1}(\rho\|\sigma) = 1- Q_2(\rho\|\sigma)$.

The second part of the error comes from the variational algorithm for estimating the standard quantum $f$-divergence,
\begin{align}
    \left|D_{-r_m}(\rho\|\sigma) - \hat{D} (\rho\|\sigma)\right| & = \left|\sum_{j=1}^m \frac{w_j}{\ln2}\left(D_{f_{t_j}}(\rho\|\sigma) - \hat{D}_{f_{t_j}}(\rho\|\sigma) \right)\right|\\
    & \le  \sum_{j=1}^m \frac{w_j}{\ln2}\varepsilon_V D_{f_{t_j}}(\rho\|\sigma)\\
    & = \varepsilon_V D_{-r_m}(\rho\|\sigma)\\
    &  = \varepsilon_V \left(D_{-r_m}(\rho\|\sigma) - D(\rho\|\sigma) + D(\rho\|\sigma) \right) \\
    &  \le \varepsilon_V\left(\frac{Q_0(\rho\|\sigma) + Q_2(\rho\|\sigma)}{m^2\ln2} + D(\rho\|\sigma)\right),
\end{align}
where the first inequality follows from Eq.~\eqref{eq:epsilon_V_assum} and $w_j > 0$ for GRJ quadrature, the last inequality follows from Eq.~\eqref{eq:error_1}. Then the total relative error becomes
\begin{align}
    \left|\frac{\hat D(\rho\|\sigma) - D(\rho\|\sigma)}{D(\rho\|\sigma)}\right| &\le \frac{D_{-r_m}(\rho\|\sigma) - D(\rho\|\sigma) + \left|D_{-r_m}(\rho\|\sigma) - \hat{D} (\rho\|\sigma)\right|}{D(\rho\|\sigma)} \\
    &\le \frac{Q_0(\rho\|\sigma) + Q_2(\rho\|\sigma)}{m^2 D(\rho\|\sigma)\ln2}(1+\varepsilon_V) + \varepsilon_V,
\end{align}
which completes the proof.
\end{proof}

\subsection{Error bound for estimating Petz \Renyi divergence}
The above analysis similarly applies to the Petz \Renyi divergence. The error $\ve_Q$ arises from the difference between the quadrature approximation $h_m(x)$ and the exact function $h(x) = x^{1-\alpha}$. This difference can be bounded by using Theorem 1 and Theorem 9 from~\cite{faust2023rational}.

\begin{lemma}[\cite{faust2023rational}]
    Consider the function $h(x) = x^{1-\alpha}$, whose integral representation is $h(x) = 1+\int_0^1 f_t(x) \md \nu_\alpha(t)$,
    where $\alpha \in (0,1)\cup (1,2)$, $f_t(x)$ is defined in Eq.~\eqref{def:f_t(x)}, and $\md \nu_\alpha(t) = \frac{\sin(\alpha\pi)}{\pi} t^{\alpha -1 }(1-t)^{1-\alpha}\md t.$
    Let $h_m(x)$ be the GRJ quadrature approximations to $h(x)$ with the fixed node at $1$. Then for $m\ge 1$ and $x>0$, we have
    \begin{align}
        0\le \frac{h(x) - h_m(x) }{1-\alpha}\le \alpha\nu_{\alpha,m}\frac{(x-1)^2}{x},
        \label{eq:h(x)-h_m(x)}
    \end{align}
    where $\nu_{\alpha,m} = \max\left\{\nu_{\alpha,m}^0, \nu_{\alpha,m}^1 \right\}$ and
    \begin{align}
        \nu_{\alpha,m}^0 = O\left(\frac{\Gamma(\alpha)}{m^{2\alpha}\Gamma(2-\alpha)}\right), \qquad \nu_{\alpha,m}^1 = O\left(\frac{\Gamma(2-\alpha)}{m^{2(2-\alpha)}\Gamma(\alpha)}\right).
    \end{align}
    We note here that for $\alpha\in (0,1)$, $h(x) - h_m(x)\ge 0$; while for $\alpha\in (1,2)$, $h(x) - h_m(x)\le 0$. When $m\to \infty$, $h_m(x)$ converges to $h(x)$.
    \label{lemma:Err_bound_fawzi}
\end{lemma}
Since we obtain the estimate of $D_\alpha(\rho\|\sigma)$ via $D_\alpha(\rho\|\sigma)=\frac{1}{\alpha-1}\log Q_\alpha(\rho\|\sigma)$, we first provide the error bound for estimating $Q_\alpha(\rho\|\sigma)$ as follows.
\begin{lemma}[Error bound for estimating $Q_\alpha(\rho\|\sigma)$]
Let $\alpha\in (0,1)\cup (1,2)$. For any quantum states $\rho$ and $\sigma$ with $\supp(\rho) \subseteq \supp(\sigma)$, if the relative error of Algorithm~\ref{alg:algorithm_D_ft} is within $\varepsilon_V$ (as defined in Eq.~\eqref{eq:epsilon_V_assum})
, then the error of estimating $Q_\alpha(\rho\|\sigma)$ is bounded as follows: 
\begin{align}
\left|{Q_\alpha(\rho\|\sigma)} - {\hat{Q}_\alpha(\rho\|\sigma)}\right| \le &C_1 \nu_{\alpha,m}(1+\ve_V) + C_2 \varepsilon_V,
\end{align}
where $C_1 := 2\left|\alpha(1-\alpha)\right|\cdot (Q_2(\rho\|\sigma)-1)$, $C_2 := \left|Q_\alpha(\rho\|\sigma) - 1\right|$ and $\nu_{\alpha,m}$ defined in Lemma~\ref{lemma:Err_bound_fawzi}.
\label{lemma:Error_Q_alpha}
\end{lemma}
\begin{proof}
    The total error stems from two parts, i.e., $|Q_\alpha(\rho\|\sigma) - \hat{Q}_{\alpha}(\rho\|\sigma)|\le |Q_\alpha(\rho\|\sigma) - D_{h_m}(\rho\|\sigma)| + |D_{h_m}(\rho\|\sigma) - \hat{Q}_\alpha(\rho\|\sigma)|$. Since $Q_\alpha(\rho\|\sigma)$ generates from $x^{1-\alpha}$-divergence, the first part of the error stems from the difference between $x^{1-\alpha}$ and $h_m(x)$. Due to Eq.~\eqref{eq:h(x)-h_m(x)}, we have
    \begin{align}
        \left|D_{x^{1-\alpha}}(\rho\|\sigma) - D_{h_m}(\rho\|\sigma)\right| \le& \left|\alpha(1-\alpha)\right| \nu_{\alpha,m} D_{\frac{1}{x}+x-2}(\rho\|\sigma) \\
        \le & \left|\alpha(1-\alpha)\right| \nu_{\alpha,m}\cdot (Q_2(\rho\|\sigma)-1),
        \label{eq:err_trace_petz_1}
    \end{align}
    where the last inequality follows from $Q_0(\rho\|\sigma)\le 1$. Since the relative error of Algorithm~\ref{alg:algorithm_D_ft} is bounded by $\varepsilon_V$, we have
    \begin{align}
        \left|\hat{Q}_\alpha(\rho\|\sigma) - D_{h_m}(\rho\|\sigma)\right| =& \left| \frac{\sin(\alpha\pi)}{\pi}\sum_{j=1}^m w_j \left(\hat{D}_{f_{t_j}}(\rho\|\sigma) - D_{f_{t_j}}(\rho\|\sigma)\right)\right| \\
        \le & \left|\frac{\sin(\alpha\pi)}{\pi}\right|\sum_{j=1}^m w_j \varepsilon_V D_{f_{t_j}}(\rho\|\sigma) \nonumber\\
        = & \varepsilon_V \left|D_{h_m}(\rho\|\sigma) - 1\right|\label{eq: Qhat Dhm tmp} \\
        \le & \varepsilon_V \left(\left|Q_\alpha(\rho\|\sigma) - 1\right|+\left|\alpha(1-\alpha)\right|\nu_{\alpha,m} \left(Q_2(\rho\|\sigma) - 1\right)\right),
        \label{eq:err_trace_petz_2}
    \end{align}
    where the last inequality follows from Eq.~\eqref{eq:err_trace_petz_1}. Combining Eqs.~\eqref{eq:err_trace_petz_1} and~\eqref{eq:err_trace_petz_2}, we obtain the asserted result.
\end{proof}
This lemma implies that when $m\to \infty$, the error stems from quadrature approximation vanishes as $\nu_{\alpha,m}\to 0$. 
It also leads to the error analysis for $D_\alpha(\rho\|\sigma)$ as follows.

\begin{proposition}[Error bound for estimating $D_\alpha(\rho\|\sigma)$]
\label{pro:error_petz_alpha}
Let $\alpha\in (0,1)\cup (1,2)$. For any quantum states $\rho$ and $\sigma$ with $\supp(\rho) \subseteq \supp(\sigma)$, if the relative error of Algorithm~\ref{alg:algorithm_D_ft} is within $\varepsilon_V$ (as defined in Eq.~\eqref{eq:epsilon_V_assum}), 
then the total relative error of Algorithm~\ref{alg:algorithm_Petz_a} can be bounded by
\begin{align}
    \left|\frac{D_{\alpha}(\rho\|\sigma) - \hat{D}_{\alpha}(\rho\|\sigma)}{D_\alpha(\rho\|\sigma)}\right| \le O(\nu_{\alpha,m}) + O\left(\varepsilon_V\right),
\end{align}
where $\nu_{\alpha,m}$ is defined in Lemma~\ref{lemma:Err_bound_fawzi} and $m$ is the number of quadrature nodes.
\end{proposition}
\begin{proof}
    The total error for estimating the Petz \Renyi divergence is bounded by two logarithmic terms, i.e.,
    \begin{align}
        \left|D_{\alpha}(\rho\|\sigma) - \hat{D}_{\alpha}(\rho\|\sigma)\right|&\le\left|\frac{1}{\alpha - 1}\log\frac{Q_\alpha(\rho\|\sigma)}{D_{h_m}(\rho\|\sigma)}\right|+\left|\frac{1}{\alpha-1}\log\frac{D_{h_m}(\rho\|\sigma)}{\hat{Q}_\alpha(\rho\|\sigma)}\right|.
    \end{align}
    For the ease of presentation, we bound them with the $O$ notation as follows,
    \begin{align}
        \left|\frac{1}{\alpha-1}\log\frac{Q_\alpha(\rho\|\sigma)}{D_{h_m}(\rho\|\sigma)}\right| =& O\left(\left| \frac{1}{\alpha-1}\cdot \frac{Q_\alpha(\rho\|\sigma) - D_{h_m}(\rho\|\sigma)}{Q_\alpha(\rho\|\sigma)} \right|\right) \\
        \le & O\left(\alpha\nu_{\alpha,m}\frac{Q_2(\rho\|\sigma)-1}{Q_\alpha(\rho\|\sigma)}\right),
        \label{eq:err_term1_Petz_alpha}
    \end{align}
    where the equality follows as $\log x \leq x-1$ and the inequality follows from Eq.~\eqref{eq:err_trace_petz_1}. The second term can be bounded by
    \begin{align}
        \left|\frac{1}{1-\alpha}\log\frac{D_{h_m}(\rho\|\sigma)}{\hat{Q}_\alpha(\rho\|\sigma)}\right| = & \left|\frac{1}{1-\alpha}\log\left(1 + \frac{\hat{Q}_\alpha(\rho\|\sigma) - D_{h_m}(\rho\|\sigma)}{D_{h_m}(\rho\|\sigma)}\right)\right| \\
        = & O\left(\left|\frac{\hat{Q}_\alpha(\rho\|\sigma) - D_{h_m}(\rho\|\sigma)}{(1-\alpha)D_{h_m}(\rho\|\sigma)}\right|\right) \\
        \le & O\left(\varepsilon_V\frac{D_{h_m}(\rho\|\sigma) - 1}{(\alpha-1)D_{h_m}(\rho\|\sigma)}\right) \\
        \le &O\left(\frac{\varepsilon_V}{|1-\alpha|}\max\left\{1,\frac{1}{D_{h_2}(\rho\|\sigma)}-1\right\} \right),
        \label{eq:err_term2_Petz_alpha}
    \end{align}
    where the first inequality follows from Eq.~\eqref{eq: Qhat Dhm tmp} and the second inequality follows from two properties of $D_{h_m}(\rho\|\sigma)$: (a) for $\alpha\in (1,2)$, $D_{h_m}(\rho\|\sigma) \ge 1$ and (b) for $\alpha\in(0,1)$ and $m\ge 2$, $D_{h_m}(\rho\|\sigma)\le D_{h_{m+1}}(\rho\|\sigma)\le 1$~\cite{faust2023rational}. Combining Eqs.~\eqref{eq:err_term1_Petz_alpha} and~\eqref{eq:err_term2_Petz_alpha}, the total relative error can be bounded by $|D_{\alpha}(\rho\|\sigma) - \hat{D}_{\alpha}(\rho\|\sigma)|/D_\alpha(\rho\|\sigma)\le O(\nu_{\alpha,m}) + O\left(\varepsilon_V\right)$.
\end{proof}

In the case of $\alpha = 2$, the error arises solely from the heuristic nature of the variational process, because the value of $Q_2(\rho\|\sigma)$ can be exactly determined from the value of $D_{f_1}(\rho\|\sigma)$ due to Eq.~\eqref{eq:approx_tr_petz2}. Therefore, if the relative error of Algorithm~\ref{alg:algorithm_D_ft} can be controlled within $\varepsilon_V$, the error bound for estimating $D_2(\rho\|\sigma)$ is given by:
\begin{align}
    \left| \frac{D_2(\rho\|\sigma) - \hat{D}_2(\rho\|\sigma)}{D_2(\rho\|\sigma)} \right| &= \frac{1}{D_2(\rho\|\sigma)}\left|\log \left(1-D_{f_1}(\rho\|\sigma)\right) - \log \left(1- \hat{D}_{f_1}(\rho\|\sigma)\right)\right| \\
    &= \frac{1}{D_2(\rho\|\sigma)}\left|\log \left(1 + \frac{D_{f_1}(\rho\|\sigma) - \hat{D}_{f_1}(\rho\|\sigma)}{\tr[\rho^2\sigma^{-1}]}\right)\right| \\
    &= O\left(\frac{\left|D_{f_1}(\rho\|\sigma) - \hat{D}_{f_1}(\rho\|\sigma)\right|}{Q_2(\rho\|\sigma)D_2(\rho\|\sigma)}\right) \\
    &\le O\left(\varepsilon_V\cdot \frac{Q_2(\rho\|\sigma)-1}{Q_2(\rho\|\sigma)D_2(\rho\|\sigma)}\right) = O\left(\varepsilon_V \right),
\end{align}
which is proportional to $\varepsilon_V$.

\section{Numerical simulations}
\label{sec:numerical_simu}

In this section, we present numerical simulations to validate the effectiveness of our variational algorithms via the PennyLane package~\cite{bergholm2018pennylane}. We begin by outlining our methodology and parameter settings, followed by a discussion of the simulation results.

\subsection{Numerical setup}

\noindent 
\textbf{Preparing input quantum states.} We consider scenarios where the input states are one-qubit and two-qubit mixed states. For this, we use the interface \textit{qml.QubitStateVector} from PennyLane to generate two pure quantum states $\ket{\rho}_{RS}$ and $\ket{\sigma}_{RS}$. These states effectively produce two mixed states on the marginal system $S$, that is, $\rho_S = \tr_R[\ket{\rho}\bra{\rho}]$ and $\sigma_S = \tr_R[\ket{\sigma}\bra{\sigma}]$. 

\vspace{0.3cm}
\noindent
\textbf{Parameterized quantum circuits.} Our ansatz is shown in Figure~\ref{fig:visual_ansatz}. For one-qubit cases, we set $U_{\bm{\theta}}$ and $V_{\bm{\beta}}$ as the generic single-qubit rotation gate $U_3$ in Figure~\ref{fig:visual_ansatz}(a). For two-qubit cases, we set $U_{\bm{\theta}}$ and $V_{\bm{\beta}}$ as the ansatz in Figure~\ref{fig:visual_ansatz}(b) with different parameters. We use $4$ layers of the two-qubit ansatz. 
To evaluate the gradient of the parameters in the parameterized quantum circuit, we utilize the function \textit{qml.gradients.param\_shift} in PennyLane. 

\begin{figure}[htbp]
    \centering
    \subfloat[Single-qubit]{\includegraphics[width=0.15\linewidth]{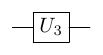}}\qquad\qquad \quad
    \subfloat[Two-qubit]{\includegraphics[width=0.28\linewidth]{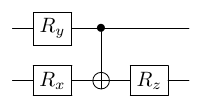}}
    \caption{Visualization of a single layer of our ansatz. (a) In single-qubit cases, the ansatz is a generic single-qubit rotation gate; (b) In two-qubit cases, the ansatz is constructed with three single-qubit rotation gates around each axis and a CNOT gate that generates entanglement.}
    \label{fig:visual_ansatz}
\end{figure}

\vspace{0.3cm}
\noindent
\textbf{Parameter setting.} We estimate the values in Eq.~\eqref{eq:sample_p} with sample means. The number of samples is set to be $10000$. The number of quadrature nodes $m$ is taken at $6$. The learning rate $\ell_r$ and the number of iterations $K$ depend on specific situations. The choices of these parameters are empirical. In particular, if $m$ is required to be large, a smaller learning rate $\ell_r$ and a larger number of iterations $K$ are required to ensure the precision of estimation. An adaptive learning rate strategy for large $m$ is shown in Appendix~\ref{sec: Training with adaptive learning rate}. The final loss is determined as the average of the loss values from the last $10$ iterations.

\subsection{Numerical results}

The numerical results of our simulations are shown in Table~\ref{tab:Num_results}. For one-qubit simulations, we set the iteration number $K$ as $300$; while for two-qubit simulations, we set $K$ as $200$. The loss of the first $100$ iterations of each case is shown in Figure~\ref{fig:conveg_Rel_qubit1}--\ref{fig:conveg_Petz_qubit2}. For all the cases, we set the learning rate $\ell_r$ as $0.1$. We choose $\alpha = 1.5$ as an example for Petz-\Renyi divergence. We also note that the Petz-$2$ \Renyi divergence $D_2(\rho\|\sigma)$ does not require a quadrature approximation, as it can be directly obtained from $D_{f_1}(\rho\|\sigma)$.

All exact and estimated values in the table are retained to four decimal places. It is easy to see that the quadrature approximations match the corresponding exact values to four decimal places and the relative errors of our estimated values are sufficiently small, most of which are within $1\%$. 

\setlength\extrarowheight{2pt}
\begin{table}[ht]
\centering
\small
\begin{tabular}{cc|c|c|c|c}
\toprule[2pt]
\multicolumn{2}{c|}{\textbf{Quantum Relative Entropies}}   & \makecell{\textbf{Exact Value}} & \textbf{Quadratures} & \makecell{\textbf{Estimated Value}} & \makecell{\textbf{Relative Error}} \\ \hline
\multicolumn{1}{c|}{\multirow{3}{*}{\makecell{$1$-qubit \\ cases}}} &  $D(\rho\|\sigma)$ &  $0.2470$  &        $0.2470$    &  $0.2455$ &  $0.59\%$ \\ \cline{2-6} 
\multicolumn{1}{c|}{}     & $D_{1.5}(\rho\|\sigma)$  & $0.3381$  &  $0.3381$     & $0.3314$ &  $2.01\%$     \\ \cline{2-6} 
\multicolumn{1}{c|}{}     & $D_2(\rho\|\sigma)$ &  $0.4073$  &   $\backslash$  & $0.4096$  &   $0.56\%$    \\ \hline
\multicolumn{1}{c|}{\multirow{3}{*}{\makecell{$2$-qubit \\ cases}}} & $D(\rho\|\sigma)$  &  $0.6585$  &        $0.6585$  & $0.6514$  &  $1.07\%$  \\ \cline{2-6} 
\multicolumn{1}{c|}{}  & $D_{1.5}(\rho\|\sigma)$ &   $0.9308$   &    $0.9308$  &  $0.9244$  & $0.70\%$   \\ \cline{2-6} 
\multicolumn{1}{c|}{}   &  $D_2(\rho\|\sigma)$ &  $1.1430$  &    $\backslash$    &  $1.1377$ & $0.46\%$        \\ \bottomrule[2pt]
\end{tabular}
\caption{A table of numerical results that compares exact values with their estimated values.
}
\label{tab:Num_results}
\end{table}

\begin{figure}[H]
    \centering
    \includegraphics[width=\textwidth]{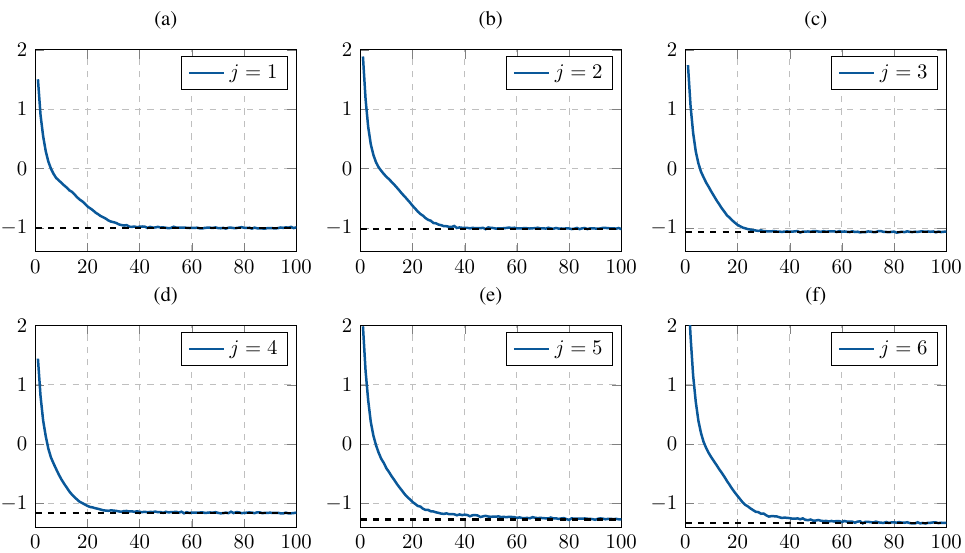}
    \caption{One-qubit case of estimating $D(\rho\|\sigma)$, where $j$ represents the index of the quadrature nodes $\{t_j\}_{j=1}^6$. The horizontal axis represents the number of iterations, and the vertical axis gives the loss value. The dashed line indicates the exact loss value, while the blue line is the loss value during training.}
    \label{fig:conveg_Rel_qubit1}
\end{figure}

\begin{figure}[H]
    \centering
    \includegraphics[width=\textwidth]{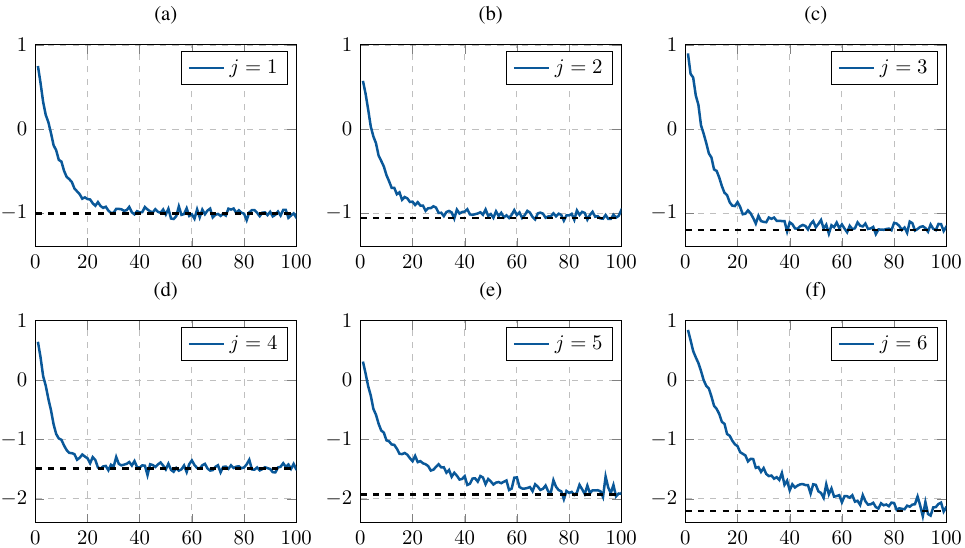}
    \caption{Two-qubit case of estimating $D(\rho\|\sigma)$, where $j$ represents the index of the quadrature nodes $\{t_j\}_{j=1}^6$. The horizontal axis represents the number of iterations, and the vertical axis gives the loss value. The dashed line indicates the exact loss value, while the blue line is the loss value during training.}
    \label{fig:conveg_Rel_qubit2}
\end{figure}

\begin{figure}[H]
    \centering
    \includegraphics[width=\textwidth]{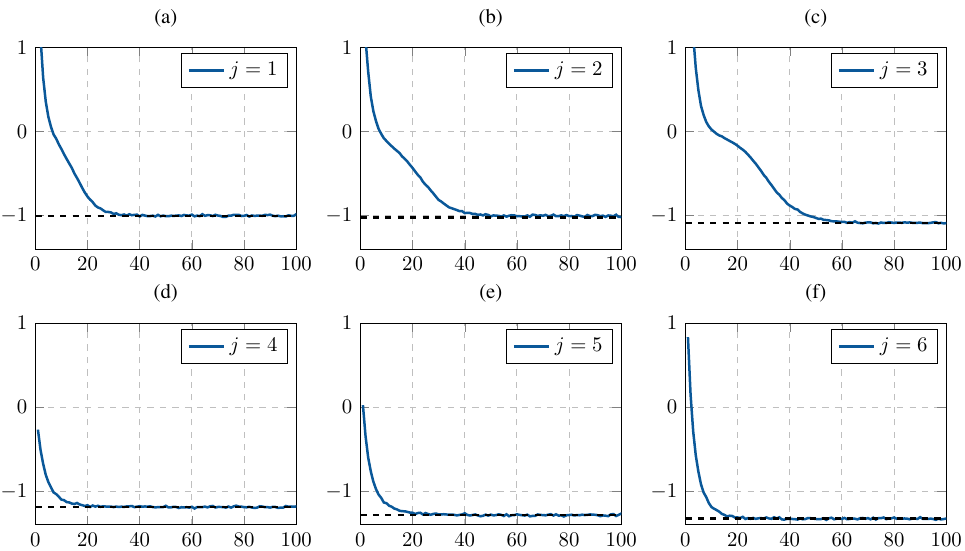}
    \caption{One-qubit case of estimating $D_\alpha(\rho\|\sigma)$ and $D_2(\rho\|\sigma)$, where $j$ represents the index of the quadrature nodes $\{t_j\}_{j=1}^6$. The horizontal axis represents the number of iterations, and the vertical axis gives the loss value. The dashed line indicates the exact loss value, while the blue line is the loss value during training. $D_2(\rho\|\sigma)$ is obtained with the final loss of $j=6$ as the quadrature node $t_6 = 1$.}
    \label{fig:conveg_Petz_qubit1}
\end{figure}

\begin{figure}[H]
    \centering
    \includegraphics[width=\textwidth]{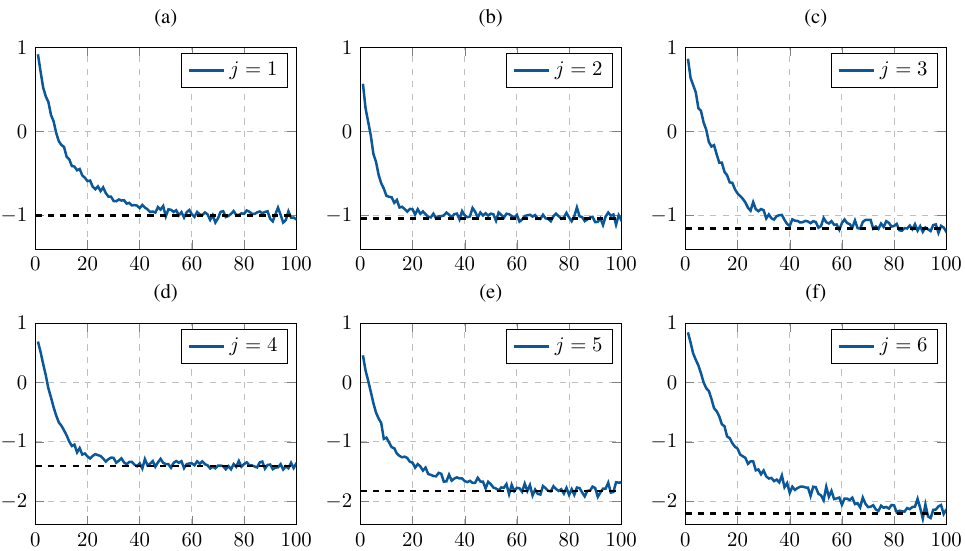}
    \caption{Two-qubit case of estimating $D_\alpha(\rho\|\sigma)$ and $D_2(\rho\|\sigma)$, where $j$ represents the index of the quadrature nodes $\{t_j\}_{j=1}^6$. The horizontal axis represents the number of iterations, and the vertical axis gives the loss value. The dashed line indicates the exact loss value, while the blue line is the loss value during training. $D_2(\rho\|\sigma)$ is obtained with the final loss of $j=6$ as the quadrature node $t_6 = 1$.}
    \label{fig:conveg_Petz_qubit2}
\end{figure}

\section{Application in quantum channel capacity}
\label{sec:applications}

As an application, we utilize our algorithm to explore the superadditivity of quantum channel capacity—a phenomenon where the combined use of quantum channels exceeds the sum of their individual capacities. Specifically, we propose a variational approach in Section~\ref{sec:superadditivity} to compute the quantum channel coherent information, employing our algorithm for quantum relative entropy as a subroutine. Numerical simulations in Section~\ref{sec:supperadd_numerical} uncover new examples of qubit channels exhibiting strict superadditivity of coherent information.

\subsection{Theoretical framework}
\label{sec:superadditivity}

Quantum channels exhibit several fascinating properties that distinguish them from their classical counterparts. One such property is the ability to transmit information superadditively, where multiple uses of the same channel can increase the total amount of information that can be reliably transmitted. Let $I(A\>B)_\sigma := D(\sigma_{AB}\|I_A \ox \tr_A(\sigma_{AB}))$ denote the coherent information of a quantum state $\sigma_{AB}$, where $I_A$ is the identity operator. The well-known quantum capacity theorem~\cite{devetak2005private, lloyd1997capacity, shor2002quantum} states that the quantum capacity $Q(\cN)$ of a quantum channel is given by the regularized channel coherent information:
\begin{align}\label{eq: channel capacity theorem}
    Q(\cN) = \lim_{k\to \infty}\frac{1}{k}I(\cN^{\ox k}),
\end{align}
where 
\begin{align}
    I(\cN) \equiv \sup_{\phi_{AA'}}I(A\>B)_{\cN_{A'\to B}(\phi_{AA'})}
\end{align}
is the channel coherent information, and the optimization is performed over all pure, bipartite quantum states $\phi_{AA'}$. Due to the superadditivity of channel coherent information, the limit in Eq.~\eqref{eq: channel capacity theorem} is generally necessary. 

An important question in quantum Shannon theory is to identify channels that exhibit strict superadditivity of coherent information, i.e., $I(\cN^{\ox n}) > nI(\cN)$. This phenomenon indicates that entangled channel coding or error correction strategies can outperform non-entangled ones, which has profound implications for quantum communication and quantum error correction. Existing studies have addressed this problem by classical neural networks and optimizations~\cite{leditzky2018dephrasure,bausch2020quantum,sidhardh2022exploring}. In this work, we take a different approach and propose a variational quantum algorithm to estimate the quantum channel coherent information directly on quantum computers.

\begin{figure}[ht]
    \centering
    \includegraphics[width=\linewidth]{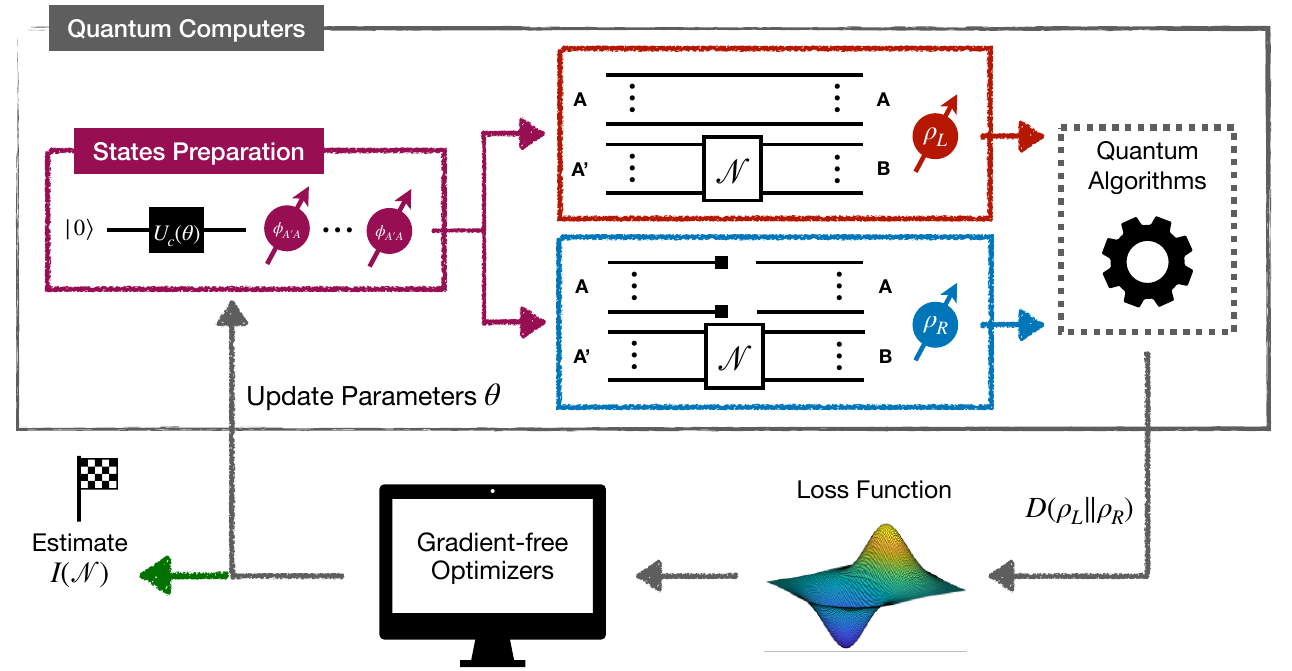}
    \caption{Estimating channel coherent information within the variational framework. Using a parameterized quantum circuit $U_c(\bm{\theta})$ and the target quantum channel $\cN$, we prepare the states $\rho_L$ and $\rho_R$ and estimate their quantum relative entropy using Algorithm~\ref{alg:algorithm_rel}. The circuit parameters are then updated using a gradient-free optimizer. Once the optimization process is complete, $D(\rho_L\|\rho_R)$ is returned as an estimate of $I(\cN)$.}
    \label{fig:superadd}
\end{figure}

First of all, the quantum channel coherent information can be expressed as:
\begin{align}
    I(\cN) &= \sup_{\phi_{AA'}} I(A\>B)_\sigma
    = \sup_{\phi_{AA'}} D(\sigma_{AB}\|I_A \ox \tr_A[\sigma_{AB}]),
\end{align}
where $\sigma_{AB} = \cN_{A'\to B}(\phi_{AA'})$. Let $\cR_A^I = \tr_A[\cdot]I_A$ represent the replacement channel on system $A$. Then, the channel coherent information can be rewritten as:
\begin{align}
    I(\cN) = \sup_{\phi_{AA'}} D(\cN_{A'\to B}(\phi_{AA'})\|\cR_A^I \ox \cN_{A'\to B}(\phi_{AA'})).
    \label{eq:opt_superadd}
\end{align}
Using the quantum algorithm for estimating quantum relative entropy, we can solve the supremum optimization problem in Eq.~\eqref{eq:opt_superadd} within a variational framework. As illustrated in Figure~\ref{fig:superadd}, the process begins with the input state $\ket{\bm{0}}$. A parameterized quantum circuit $U_c(\bm{\theta})$ is applied to prepare the state $\ket{\phi(\bm{\theta})} = U_c(\bm{\theta}) \ket{\bm{0}}$. The quantum channel $\cN$ is then applied to generate the states $\rho_L$ and $\rho_R$, defined as:
\begin{align}
    \rho_L = \cN_{A'\to B}(\phi(\bm{\theta})), \quad \rho_R = \cR_A^I \ox \cN_{A'\to B}(\phi(\bm{\theta})).
    \label{eq:rho_LR}
\end{align}
The states $\rho_L$ and $\rho_R$ are used as inputs to the quantum algorithm for estimating quantum relative entropy. The parameters $\bm{\theta}$ are updated iteratively using a gradient-free optimizer, such as a genetic algorithm. This is because the coherent information landscape is dominated by local maxima, making traditional optimization methods like gradient descent less effective~\cite{bausch2020quantum}. The genetic algorithm used here follows the approach in~\cite{sidhardh2022exploring}. Once the optimization process converges, the algorithm outputs $D(\rho_L\|\rho_R)$ as an estimate of the channel coherent information $I(\cN)$.

\subsection{Numerical experiments}
\label{sec:supperadd_numerical}

We conduct numerical experiments to explore superadditive quantum channels. 
For this, we identify a quantum channel with a readily computable one-shot coherent information $I(\cN)$ and estimate its $n$-shot coherent information $I(\cN^{\ox n})$ using the method outlined in Section~\ref{sec:superadditivity}.
According to Eq.~\eqref{eq:exact_lowerbound}, our estimate $\hat{I}(\cN^{\ox n})$ always serves as a lower bound for $I(\cN^{\ox n})$, i.e., $I(\cN^{\ox n}) \geq \hat{I}(\cN^{\ox n})$.
This implies that if our estimated value $\hat{I}(\cN^{\ox n})$ is strictly greater than $nI(\cN)$, then superadditivity must exist, as $I(\cN^{\ox n}) \geq \hat{I}(\cN^{\ox n}) > n I(\cN)$.

In particular, we choose the set of three-parameter family of Pauli channels, defined as
\begin{align}
    \cP(\rho) = p_0 \sigma_0 \rho \sigma_0 + p_1 \sigma_1 \rho \sigma_1 + p_2 \sigma_2 \rho \sigma_2 + p_3 \sigma_3 \rho \sigma_3, 
\end{align}
where $p_i \ge 0$, $\sum_{i=0}^3 p_i = 1$, and $\{\sigma_i\}_{i=0}^3$ is the set of Pauli matrices. It is well known that the coherent information of a Pauli channel is~\cite{bennett1996mixed,bausch2021error} 
\begin{align}
    I(\cP) = 1 + \sum_{i=0}^3 p_i \log p_i.
\end{align} 
Set $n = 2$ and the loss function as
\begin{align}
    L(\bm{p}) = \frac{1}{2}D(\rho_L\|\rho_R) - I(\cP),
\end{align}
where $\rho_L$ and $\rho_R$ are defined in Eq.~\eqref{eq:rho_LR}, and $\bm{p} = (p_0, p_1, p_2,p_3)$. 
The superadditivity is evident when we observe $L(\bm{p}) > 0$.

\vspace{0.3cm}
\noindent
\textbf{Experimental Setup.} We discretize the channel parameters $(p_1, p_2, p_3) \in [0, 0.2]^3$ with a step size of $0.01$, resulting in a total of 8000 Pauli channels. To construct the parameterized quantum circuit $U_c(\bm{\theta})$, we use a 2-layer complex entangled ansatz, as shown in Figure~\ref{fig:superadd_ansatz}, with 4 input qubits. The initial parameters $\bm{\theta}$ are randomly selected from the interval $[0, 2\pi]$. For optimization, we employ a simple genetic algorithm, available in the \textit{Global Optimization Toolbox} of MATLAB, as our gradient-free optimizer.

\begin{figure}[ht]
    \centering
    \includegraphics[width=0.25\linewidth]{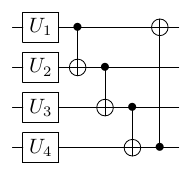}
    \caption{Complex entangled ansatz for $U_c(\bm{\theta})$. Each $U_i$ ($i\in\{1,2,3,4\}$) represents an arbitrary single-qubit rotation gate.}
    \label{fig:superadd_ansatz}
\end{figure}

\vspace{0.3cm}
\noindent
\textbf{Results.} Out of the $8000$ Pauli channels analyzed, our method identifies $69$ channels exhibiting strict superadditivity. The largest superadditivity gap is approximately $\frac{1}{2} D(\rho_L\|\rho_R) - I(\mathcal{P}) \approx 0.0078$, which matches the state-of-the-art gap found for Pauli channels~\cite{sidhardh2022exploring}. Figure~\ref{fig:3D_superadd} provides a visualization of these channels and their corresponding superadditivity values. This confirms the effectiveness of our method in identifying superadditive channels and showcases the potentials of quantum computers to address quantum-native problems.

\begin{figure}[ht]
    \centering
    \includegraphics[width=0.48\linewidth]{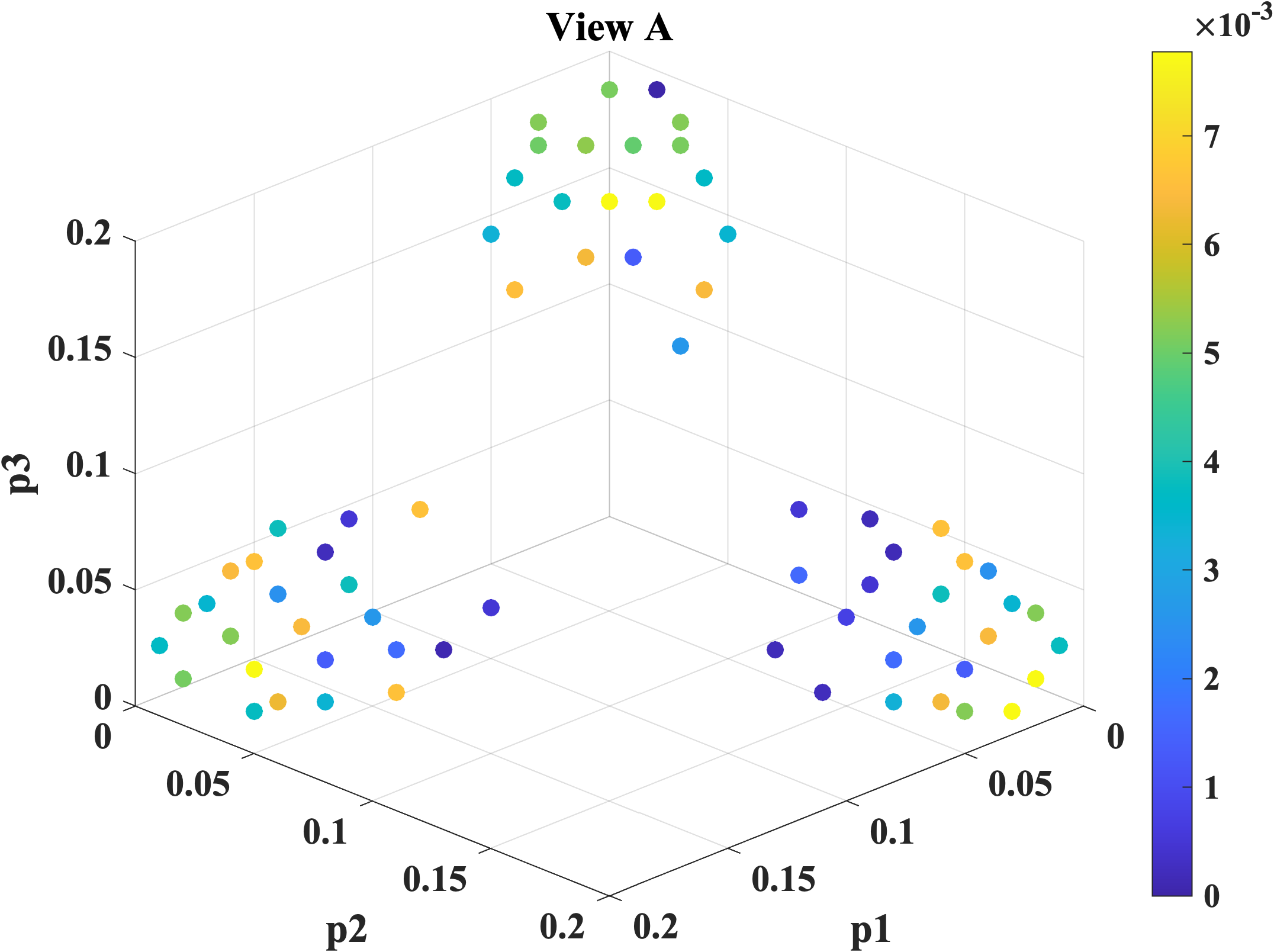}
    \includegraphics[width=0.48\linewidth]{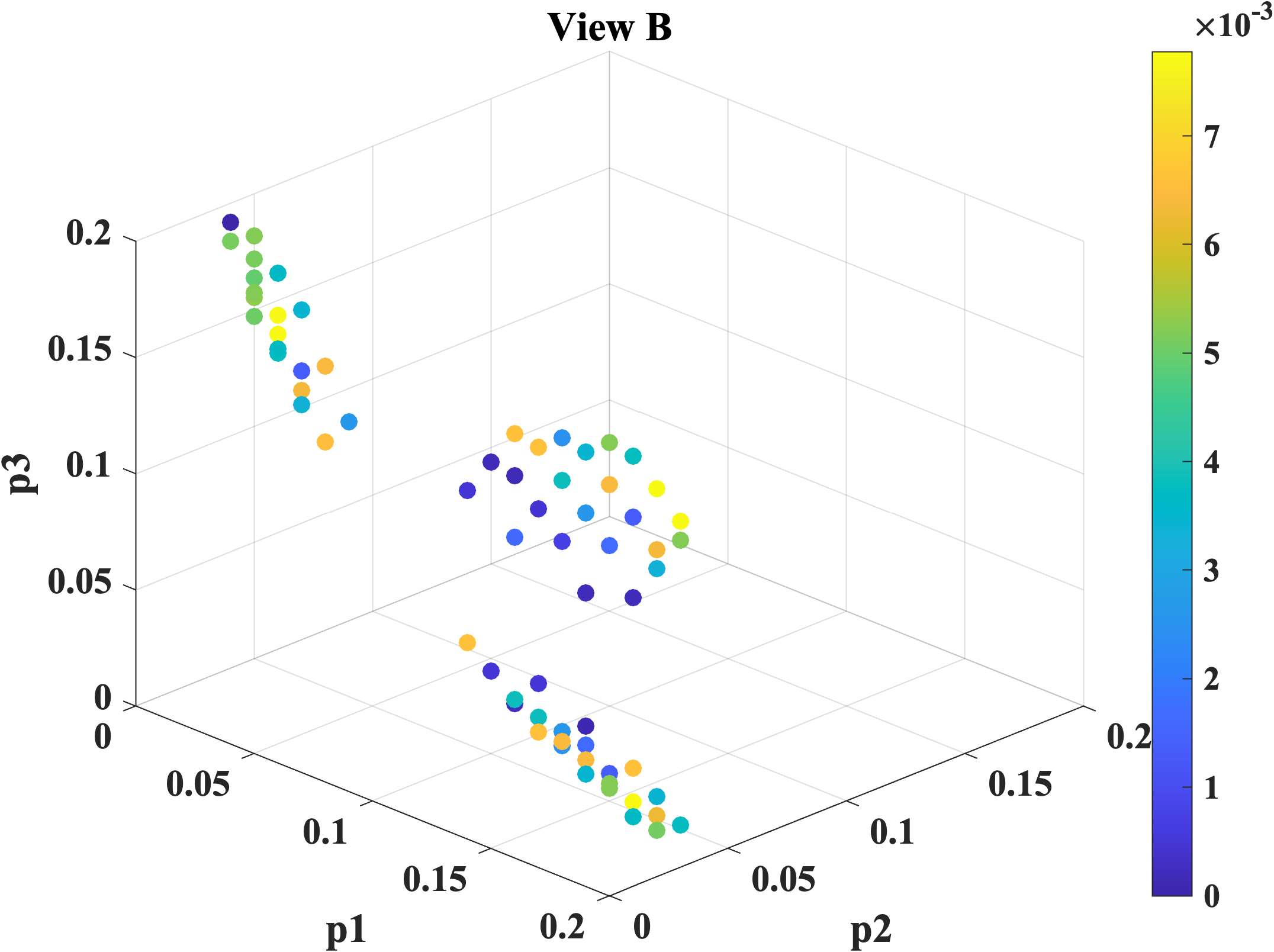}
    \caption{$3$D visualizations of Pauli channel parameters $(p_1, p_2, p_3)$, where each point represents a specific Pauli channel exhibiting superadditivity. A total of $69$ such channels are identified. The color of each point indicates the degree of superadditivity ($\frac{1}{2} D(\rho_L\|\rho_R) - I(\mathcal{P})$) as determined by our method. View A (left) shows the distribution from a front-upper perspective, while View B (right) provides a right-upper perspective.}
    \label{fig:3D_superadd}
\end{figure}

\section{Conclusions}
\label{sec:conclusion}
We have proposed a variational quantum algorithm for directly estimating quantum relative entropy and Petz \Renyi divergence between two unknown quantum states on quantum computers, addressing open problems highlighted by Goldfeld et al.~\cite{goldfeld2024quantum} and Wang et al.~\cite{wang2024new}. The feasibility and effectiveness of this approach have been validated through theoretical analysis and numerical simulations.

Efficiently computing information-theoretic quantities such as quantum relative entropy and Petz \Renyi divergence on quantum hardware represents a significant step toward realizing the full potential of quantum technologies for scientific discovery. Our demonstration of the application of this method to a well-known problem in quantum Shannon theory highlights the potential of quantum machine learning to tackle quantum-native problems and paving the way for advancing quantum science through quantum computation.

In terms of hardware implementation, a notable feature of our algorithm is its applicability to distributed quantum computing scenarios, where the quantum states to be compared are hosted on cross-platform quantum devices. This capability opens up a wide range of potential applications. Furthermore, the circuit size of our algorithm closely matches the size of the quantum states being compared, making it feasible for immediate implementation on existing quantum hardware. We leave the exploration of quantum hardware experiments as a direction for future work.

\section*{Acknowledgements}

Y.L. thanks Jinghao Ruan for providing computational support for the numerical experiments, and Zherui Chen for his helpful discussion. 
K.F. thanks Omar Fawzi for bringing his attention to the variational expression of the quantum relative entropy and He Zhang for his helpful discussions on the Gauss-Radau quadrature. 

\section*{Funding}

K.F. is supported by the National Natural Science Foundation of China (grant No. 92470113 and 12404569), the Shenzhen Science and Technology Program (grant No. JCYJ20240813113519025), the Shenzhen Fundamental Research Program (grant No. JCYJ20241202124023031), the 1+1+1 CUHK-CUHK(SZ)-GDST Joint Collaboration Fund (grant No. GRDP2025-022), and the University Development Fund (grant No. UDF01003565).

\section*{Author Contributions}

K.F. developed the main idea of this work. Y.L. derived the theoretical results and conducted the numerical simulations. Both authors contributed to the writing, discussions and interpretations of the results.

\section*{Competing Interests}

The authors declare no competing interests.

\section*{Additional Information}

After the release of our preprint~\cite{lu2025estimating}, a second version of~\cite{yao2024nonlinear} was updated, presenting a quantum phase processing approach for estimating functional quantities, including quantum relative entropy.

\bibliographystyle{ieeetr}
\bibliography{Bib}

\appendix

\section{Alternative approach to parameterizing linear operator}

\label{sec:CNN}

When parameterizing $\Lambda = \sum_{i=1}^d \lambda_i \ket{i}\bra{i}$, the set of parameters $\{\lambda_i\}_{i=1}^d$ can be approximated by a classical neural network as used in~\cite{goldfeld2024quantum}. More explicitly, let $y_{\bm{w}}: [d] \to \mathbb{R}$  be a classical neural network with a parameter vector $\bm{w}\in \mathbb{R}^s$. Then we parameterize $\Lambda = \sum_{i=1}^d y_{\bm{w}}(i) \cdot \ket{i}\bra{i}$ and the loss function becomes
\begin{align}
\cL(\bm{w}, \bm{\theta},\bm{\beta}) = \sum_{i=1}^d \left\{ t y_{\bm{w}}(i)^2 p_{\bm{\theta}}^{(i)} + (1-t)y_{\bm{w}}(i)^2 p_{\bm{\beta}}^{(i)} \right\} +(4 p_\chi -2)\tr[\Lambda].\nonumber
\end{align}
The gradients with respect to $w_j$ are
\begin{align}
\frac{\partial \cL}{\partial w_j} = \sum_{i=1}^d \left(2t y_{\bm{w}}(i) p_{\bm{\theta}}^{(i)} + 2(1-t)y_{\bm{w}}(i) p_{\bm{\beta}}^{(i)} - 2\right)\frac{\partial y_{\bm{w}}(i)}{\partial w_j} + 4\frac{\partial (p_\chi \tr[\Lambda])}{\partial w_j} ,
\end{align}
where the gradient in the last term is estimated by central difference, i.e.,
\begin{align}
\frac{\partial (p_\chi \tr[\Lambda])}{\partial w_j} \approx \frac{p_\chi(w_j + \Delta w_j)\tr[\Lambda(w_j + \Delta w_j)] - p_\chi(w_j -\Delta w_j)\tr[\Lambda(w_j - \Delta w_j)]}{2\Delta w_j}.
\end{align}
Thus, the number of queries for the extended SWAP test can be reduced from $O(d)$ to $O(s)$. Moreover, if we can prepare the diagonal quantum state $\rho_\Lambda = \Lambda/\tr[\Lambda]$ efficiently, we can estimate the following quantity with a single query of extended SWAP test,
\begin{align}
\tr[\rho(\cZ_{\bm{i}} + \cZ_{\bm{i}}^\dagger)] = \left( 4\tr[\left(\ket{0}\bra{0}\ox I\right) \chi_{\bm{i}} \left(\ket{0}\bra{0} \ox \rho\ox \rho_\Lambda \right)\chi_{\bm{i}}^\dagger] -2 \right)\cdot \tr[\Lambda],
\end{align}
which directly follows by Eq.~\eqref{eq:remark2_need1}.

\section{Alternative error analysis}
\label{appendix:error_analysis}

In this section, we provide the error analysis of our algorithm using GRJ quadrature nodes with fixed node $t_1 = 0$. Denote $\varepsilon_Q$ and $\varepsilon_V$ as the errors originating from the quadrature approximation and the heuristic nature of the variational algorithms, and $m$ as the number of quadrature nodes. Since the error originating from the heuristic nature of the variational process is not influenced by the formula of quadrature approximation, we only need to restate the error analysis for the quadrature approximation part. 

We first restate the proof of Proposition~\ref{pro:error_rel_entr} as follows.
\begin{proof}(Error bound for estimating $D(\rho\|\sigma)$).
    Let $r^0_m(x)$ and $r^1_m(x)$ be the quadrature approximations of function $\log x$ with fixed nodes at $t_1 = 0 $ and $t_m = 1$, respectively. Notice that $-r^1_m(x) = r^0_m(x^{-1})$, which leads to an inequality similar to Eq.~\eqref{eq:bound_logx}, i.e.,
    \begin{align}
        g^0(x) \le \log x\le r_m^0(x),
    \end{align}
    where $g^0(x) = r^0_m(x) - \frac{1}{m^2\ln 2}\left(\frac{1}{x} + x - 2\right)$. Due to the linearity and monotonicity of $f$-divergence, we have $D_{-g^0}(\rho\|\sigma) \ge D(\rho\|\sigma) \ge D_{-r^0_m}(\rho\|\sigma)$, which indicates that the difference between $D(\rho\|\sigma)$ and $D_{-r^0_m}(\rho\|\sigma)$ must be smaller than the difference between $D_{-r^0_m}(\rho\|\sigma)$ and $D_{-g^0}(\rho\|\sigma)$. Then we reach the same result as Eq.~\eqref{eq:error_1},
    \begin{align}
        D(\rho\|\sigma) - D_{-r^0_m}(\rho\|\sigma) &\le D_{-g^0}(\rho\|\sigma) - D_{-r^0_m}(\rho\|\sigma) \\
        &\le \frac{Q_0(\rho\|\sigma) + Q_2(\rho\|\sigma)}{m^2\ln2}.
    \end{align}
    The second part of the error that originates from the variational process remains the same, and thus the final error bound remains the same, which completes the proof.
\end{proof}

Then we restate the proof of Lemma~\ref{lemma:Error_Q_alpha} as follows.
\begin{proof}(Error bound for estimating $Q_\alpha(\rho\|\sigma)$).
    Let $h^0_m(x)$ and $h^1_m(x)$ be the quadrature approximations of function $x^{1-\alpha}$ with fixed nodes at $t_1 = 0 $ and $t_m = 1$, respectively. Based on Theorems 1 and 9 from \cite{faust2023rational}, the error of the GRJ quadrature approximation with the fixed node at $0$ is 
    \begin{align}
        0 \le \frac{h_m^0(x) - h(x)}{1-\alpha} \le \alpha \nu_{\alpha,m} \frac{(x-1)^2}{x},
        \label{eq:h_m^0(x)-h(x)}
    \end{align}
    where $\nu_{\alpha,m} = \max\left\{\nu_{\alpha,m}^0, \nu_{\alpha,m}^1 \right\}$, $\nu_{\alpha,m}^0$ and $\nu_{\alpha,m}^1$ are defined in Lemma~\ref{lemma:Err_bound_fawzi}. The total error can be bounded by $|Q_\alpha(\rho\|\sigma) - \hat{Q}_{\alpha}(\rho\|\sigma)|\le |Q_\alpha(\rho\|\sigma) - D_{h_m^0}(\rho\|\sigma)| + |D_{h_m^0}(\rho\|\sigma) - \hat{Q}_\alpha(\rho\|\sigma)|$, where the second term can still be bounded by Eq.~\eqref{eq:err_trace_petz_2}. Based on Eq.~\eqref{eq:h_m^0(x)-h(x)}, the first term can be bounded in the same way as Eq.~\eqref{eq:err_trace_petz_1}, i.e., $\left|D_{x^{1-\alpha}}(\rho\|\sigma) - D_{h_m^0}(\rho\|\sigma)\right| \le \left|\alpha(1-\alpha)\right| \nu_{\alpha,m}\cdot (Q_2(\rho\|\sigma) - 1)$, which completes the proof.
\end{proof}

Then we restate the proof of Proposition~\ref{pro:error_petz_alpha} as follows.
\begin{proof}
    (Error bound for estimating $D_\alpha(\rho\|\sigma)$). It's easy to see that when substituting $h_m(x)$ into $h_m^0(x)$, Eq.~\eqref{eq:err_term1_Petz_alpha} still holds. However, Eq.~\eqref{eq:err_term2_Petz_alpha} need to be modified as follows,
    \begin{align}
        \left|\frac{1}{1-\alpha}\log\frac{D_{h_m^0}(\rho\|\sigma)}{\hat{Q}_\alpha(\rho\|\sigma)}\right| & \le O\left(\varepsilon_V\frac{D_{h_m^0}(\rho\|\sigma) - 1}{(\alpha-1)D_{h_m^0}(\rho\|\sigma)}\right) \\
        & \le O\left(\frac{\varepsilon_V}{|1-\alpha|}\max\left\{1,\frac{1}{Q_{\alpha}(\rho\|\sigma)}-1\right\} \right),
        \label{eq:err_term3_Petz_alpha}
    \end{align}
    where the second inequality follows from two properties of $D_{h_m}(\rho\|\sigma)$: (a) for $\alpha\in (1,2)$, $D_{h_m}(\rho\|\sigma) \ge 1$ and (b) for $\alpha\in(0,1)$ and $Q_\alpha(\rho\|\sigma)\le  D_{h_m^0}(\rho\|\sigma)\le 1$. Combining Eqs.~\eqref{eq:err_term1_Petz_alpha} and~\eqref{eq:err_term3_Petz_alpha}, the total relative error can be bounded by $|D_{\alpha}(\rho\|\sigma) - \hat{D}_{\alpha}(\rho\|\sigma)|/D_\alpha(\rho\|\sigma)\le O(\nu_{\alpha,m}) + O\left(\varepsilon_V\right)$, which completes the proof.
\end{proof}

\section{Training with adaptive learning rate}
\label{sec: Training with adaptive learning rate}

In certain cases, a large number of quadrature nodes $m$ may be required to achieve the desired precision in the quadrature approximation. For illustration, we consider a particular $1$-qubit example that evalutes the Petz \Renyi divergence between $\rho = \mathrm{diag}(0.025, 0.975)$ and $\sigma = \mathrm{diag}(0.975, 0.025)$ with $\alpha = 1.5$. The exact value is $D_{\alpha}(\rho\|\sigma) = 5.2142$, with $\tr[\rho^{\alpha}\sigma^{1-\alpha}] = 6.0929$. Using $m = 6 $ quadrature nodes yields $D_{h_6} = 6.3508$ with a relative error of $4.23\%$. Increasing the quadrature nodes to $m = 16$ significantly improves the result to $D_{h_{16}} = 6.0932$, with a much smaller relative error of $4.92\times 10^{-3} \%$.

However, when estimating these quantities using our variational algorithms, we observe that as the quadrature node $ t_j $ approaches $1$, the optimization of the loss function becomes increasingly challenging. With a constant learning rate of $ \ell_r = 0.1 $, the training process succeeds for the first $12$ quadrature nodes $ \{t_j\}_{j=1}^{12} $, but fails for the last four nodes, as shown in the left column of Figure~\ref{fig:conveg_No_adapt}. It is clear that the loss function initially converges but begins to exhibit severe fluctuations midway, eventually becoming trapped in a local minimum. This behavior suggests that an excessively large learning rate prevents the algorithm from reaching the optimal solution, necessitating the use of a smaller learning rate in such cases.

To address this, we adopt an adaptive learning rate strategy to ensure both training efficiency and accuracy. We initialize the learning rate $\ell_r = 0.1$ and dynamically adjust this rate based on fluctuations in the loss function. Specifically, we monitor the loss values over the preceding $20$ iterations and perform a quadratic regression to measure the fluctuations using the mean square error of the regression. If the error exceeds $2$, the learning rate is halved, $\ell_r \leftarrow \ell_r/2$. As shown in the right column of Figure~\ref{fig:conveg_No_adapt}, this adaptive strategy effectively helps the algorithm escape the local minimum, ultimately achieving global convergence and giving an estimated value $5.1979$ with a relative error of $0.31\%$.

These results demonstrate that our algorithm can achieve high-precision estimations even when a large number of quadrature nodes is required. Furthermore, the convergence process could potentially be accelerated by employing more advanced optimization techniques, such as Nesterov accelerated gradient descent (NAGD), or by designing parameterized quantum circuits based on ansatzes with enhanced properties.

\begin{figure}[H]
    \centering
    \includegraphics[width=\textwidth]{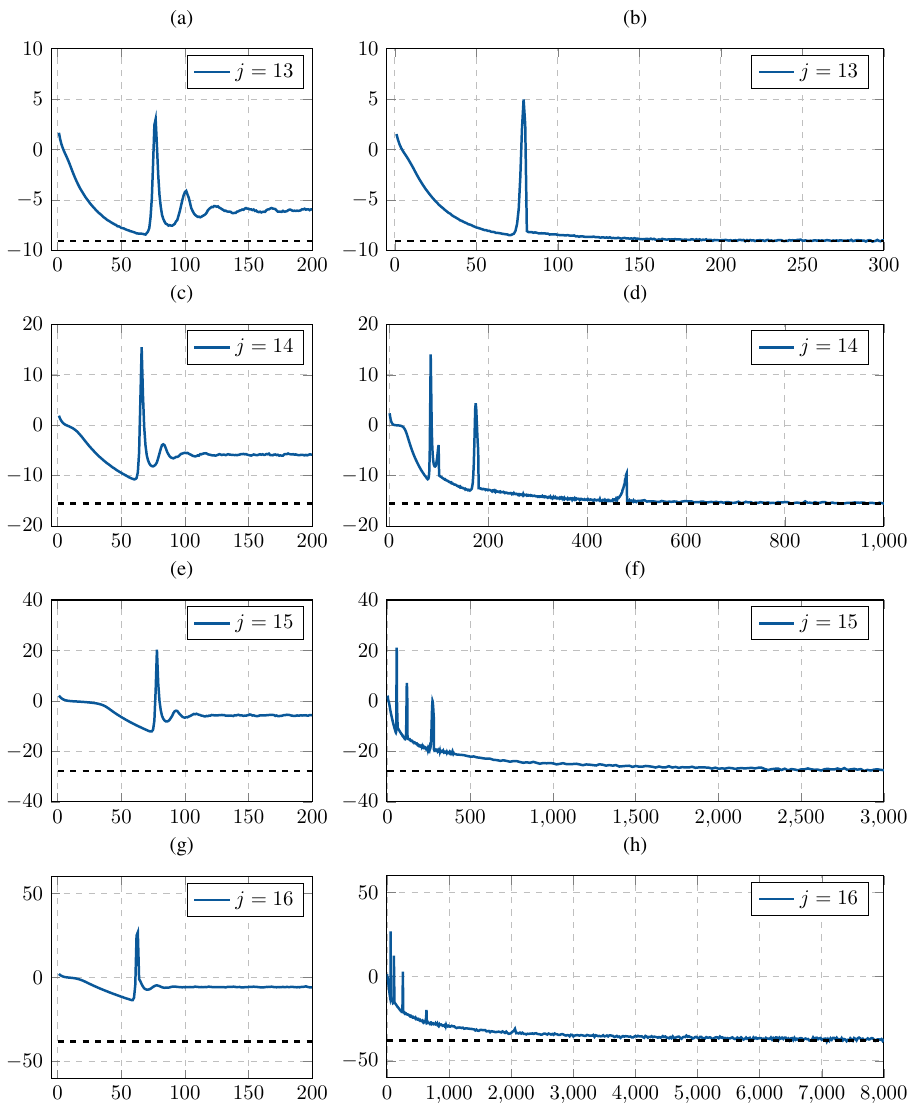}
    \caption{Training with and without adaptive learning rate. The dashed line represents the exact value of the loss, while the blue line represents the loss value during training. $j$ represents the index of the quadrature nodes. $x$-axis represents the number of iterations, while $y$-axis represents the loss value. The figures on the left show the results of the training with constant learning rate $\ell_r = 0.1$, while those on the right show the corresponding results with adaptive learning rate.}
    \label{fig:conveg_No_adapt}
\end{figure}

\end{document}